\documentclass[onecolumn,a4paper, accepted=2025-04-07]{quantumarticle}
\pdfoutput=1
\usepackage{amsmath}
\usepackage{amssymb}
\usepackage{amsthm}
\usepackage{amsfonts}
\usepackage{subcaption} 
\usepackage[table]{xcolor}
\usepackage{float}
\usepackage[colorlinks]{hyperref}
\usepackage[all]{hypcap}
\usepackage{tikz}
\usepackage{color,soul}
\usepackage[utf8]{inputenc}
\usepackage{capt-of}
\usepackage[backend=biber,style=alphabetic,maxbibnames=999]{biblatex}
\usepackage{braket}

\usepackage{graphicx}
\usepackage{thmtools}
\usepackage{mathtools}
\mathtoolsset{centercolon}
\usepackage{thm-restate}
\usepackage{comment}
\usepackage{placeins}
\usepackage{multirow}
\usepackage{dcolumn}
\usepackage{makecell}
\usepackage{enumerate}
\usepackage[shortlabels]{enumitem}
\usepackage[ruled,lined]{algorithm2e}
\usepackage[export]{adjustbox}
\usepackage{changepage}
\usepackage{listings}

\declaretheoremstyle[
  headfont=\normalfont\itshape,
  headpunct=:,
  bodyfont=\normalfont,
]{mcx_problem}

\theoremstyle{definition}
\newtheorem{definition}{Definition}[section]
\theoremstyle{definition}
\newtheorem{theorem}[definition]{Theorem}
\theoremstyle{definition}

\newcommand{\eq}[1]{\hyperref[eq:#1]{Equation~\ref*{eq:#1}}}
\renewcommand{\sec}[1]{\hyperref[sec:#1]{Section~\ref*{sec:#1}}}
\DeclareRobustCommand{\app}[1]{\hyperref[app:#1]{Appendix~\ref*{app:#1}}}
\newcommand{\fig}[1]{\hyperref[fig:#1]{Figure~\ref*{fig:#1}}}
\newcommand{\tbl}[1]{\hyperref[tbl:#1]{Table~\ref*{tbl:#1}}}
\newcommand{\problemref}[1]{\hyperref[problem:#1]{Problem~\ref*{problem:#1}}}
\newcommand{\theoremref}[1]{\hyperref[theorem:#1]{Theorem~\ref*{theorem:#1}}}
\newcommand{\definitionref}[1]{\hyperref[definition:#1]{Definition~\ref*{definition:#1}}}

\newcommand{\Google}{\affiliation{Google Quantum AI, Santa Barbara, California 93117, USA}}

\newcommand{\mcxoptcomplexityclean}{$2n - 3$ }
\newcommand{\mcxoptcomplexitydirty}{$4n - 8$ }
\newcommand{\gray}{\cellcolor{lightgray!60}}

\begin{document}

\title{Rise of conditionally clean ancillae for efficient quantum circuit constructions}

\date{February 11, 2025}

\author{Tanuj Khattar}
\email[Corresponding author: ]{tanujkhattar4@gmail.com}
\Google

\author{Craig Gidney}
\email[Corresponding author: ]{craig.gidney@gmail.com}
\Google

\begin{abstract}
We introduce conditionally clean ancilla qubits, a new quantum resource, recently explored by~\cite{nie2024quantum}, that bridges the gap between traditional clean and dirty ancillae.  
Like dirty ancillae, they begin and end in an unknown state and can be borrowed from existing system qubits, avoiding the space overhead of explicit qubit allocation.  
Like clean ancillae, they can be treated as initialized in a known state within specific computations, thus avoiding the overhead of toggle detection required for dirty ancillae.
We present new circuit constructions leveraging conditionally clean ancillae to achieve lower gate counts and depths, particularly with limited ancilla availability. Specifically, we provide constructions for:


\begin{enumerate}[label=(\alph*), noitemsep]
\item $n$-controlled NOT using $2n$ Toffolis and $\mathcal{O}(\log{n})$ depth given 2 clean ancillae.
\item $n$-qubit incrementer using $3n$ Toffolis given $\log_2^*{n}$ clean ancillae. 
\item $n$-qubit quantum-classical comparator using $3n$ Toffolis given $\log_2^*{n}$ clean ancillae.
\item unary iteration over $[0,N)$ using $2.5N$ Toffolis given $\log_2^*{n}$ clean ancillae.
\item unary iteration via skew tree over $[0, N)$ using $1.25N$ Toffolis given $n$ dirty ancillae.
\end{enumerate}

We also introduce \emph{laddered toggle detection}, a technique to replace clean ancillae with dirty ancillae in all our constructions, incurring a 2x Toffoli gate overhead.  
Our results demonstrate that conditionally clean ancillae are a valuable tool for quantum circuit design, especially in the resource-constrained early fault-tolerant era.
%

\end{abstract}

\maketitle
\emph{\textbf{Data availability:}} The circuit constructions presented in this paper are available at

\href{https://doi.org/10.5281/zenodo.12819218}{https://doi.org/10.5281/zenodo.12819218}.
All circuits are implemented as Bloqs using 

the \href{https://github.com/quantumlib/Qualtran}{Qualtran} library\cite{harrigan2024expressinganalyzingquantumalgorithms} and are tested using classical and quantum simulations.

{
  \hypersetup{linkcolor=black}
  \tableofcontents
}

\section{Introduction}

Quantum circuit design often involves trade-offs between the number of qubits used (space complexity), the number of gates required (time complexity), and the circuit depth. In the Noisy Intermediate-Scale Quantum (NISQ) era and the Early Fault-Tolerant era, minimizing qubit overhead is crucial due to the limited availability of high-quality qubits. In the fault tolerant regime, the cost of executing a circuit is often dominated by the number of non-clifford gates like T / Toffoli gates, which are significantly more expensive to execute than Clifford gates~\cite{Litinski2019, Fowler2012}. 


Several works have focused on coming up with efficient circuit constructions, with low T / Toffoli counts and reduced ancilla usage, for a wide variety of applications like quantum chemistry~\cite{PRXQuantum.4.040303, Kim2022, Lee2021}, quantum dynamics~\cite{Rubin2024, agrawal2024quantifyingfaulttolerantsimulation}, combinatorial optimization~\cite{Sanders2020}, and quantum arithmetic circuits such as for Shor's algorithm~\cite{gidney2018factoring,litinski2023compute256bitellipticcurve}.
There are often tradeoffs where one can reduce the number of T / Toffoli gates in the decomposition by using a greater number of ancilla qubits as temporary workspace~\cite{gidney2015blog, gidney2015blogincrement, Gidney2018halvingcostof, zhu2024unifiedarchitecturequantumlookup}. Ancilla qubits come in two main varieties: \emph{clean} and \emph{dirty}.

A \emph{clean} ancilla is initialized in a known state, typically $\ket{0}$, and can be discarded or reused after the computation. 
It is characterized via
\begin{itemize}[itemsep=0.001cm]
    \item[$\times$] \underline{Expensive Allocation:} Need to allocate dedicated qubits, which adds space overhead.
    \item[\checkmark]  \underline{Cheap Consumption:} Since the initial state is known, any intermediate result stored on the clean ancilla can be consumed directly via a single controlled unitary. 
    \item[\checkmark]  \underline{Cheap Uncomputation:} Using \emph{measurement-based uncomputation}~\cite{Jones2013, Gidney2018halvingcostof} whereby a clean ancilla that is entangled with other qubits can be disentangled and returned to the $\ket{0}$ state by measuring in the X basis and applying classically controlled phase corrections.
\end{itemize}

A \emph{dirty} ancilla is in an unknown, arbitrary state, that can be temporarily perturbed but must be restored eventually.
It is characterized via
\begin{itemize}[itemsep=0.001cm]
    \item[\checkmark] \underline{Cheap Allocation:} Can \emph{borrow} any system qubit not part of the decomposition.
    \item[$\times$]  \underline{Expensive Consumption:} Since the initial state is unknown, we need to use \emph{toggle detection}~\cite{gidney2015blog, nie2024quantum} to consume the intermediate result which doubles the gate count.  
    \item[$\times$]  \underline{Expensive Uncomputation:} Since the dirty ancilla must be returned to its original unknown state, we cannot use measurement based uncomputation.
 \end{itemize}

This paper describes and explores the utility of a third type of ancilla: the \emph{conditionally clean} ancilla~\cite{nie2024quantum}. 
A conditionally clean ancilla is, strictly speaking, a dirty ancilla.  It is in an unknown initial state and must be restored to the same state after use.  
However, it is \emph{conditionally} clean in the sense that, given a specific condition on other qubits in the system, the ancilla is guaranteed to be in a known state. 
This allows us to treat it \emph{as if} it were clean in many situations, avoiding the overhead associated with fully dirty ancillae. 
It is characterized via
\begin{itemize}[itemsep=0.001cm]
    \item[\checkmark] \underline{Cheap Allocation:} Can \emph{borrow} system qubits that satisfy certain criteria. 
    
    \item[\checkmark]  \underline{Cheap Consumption:} Since the initial state is conditionally known, we can consume intermediate results without the need for toggle detection. 
    
    \item[$\times$]  \underline{Expensive Uncomputation:} Since the conditionally clean ancilla must be returned to its original unknown state, we cannot use measurement based uncomputation.
 \end{itemize}

We summarize our contributions as follows:

\begin{enumerate}
    \item We formally define conditionally clean ancillae and provide a general framework for their use in quantum circuit design.
    \item We introduce laddered toggle detection, an improved technique for substituting dirty ancillae for clean ancillae in circuits where multiple clean ancillae are required.
    \item We present novel circuit constructions that leverage conditionally clean ancillae to achieve significant improvements in gate count and depth for several fundamental primitives:
    \begin{itemize}
        \item An $n$-controlled NOT (Toffoli) gate with optimal Toffoli count ($2n-3$) using only two clean ancillae and achieving $\mathcal{O}(\log n)$ depth.
        \item An $n$-qubit incrementer and an $n$-qubit quantum-classical comparator, both using $3n$ Toffoli gates and requiring only $\log_2^*{n}$ clean ancillae (where $\log_2^*{n}$ is the iterated logarithm, which is less than or equal to 5 for all practical input sizes).
        \item Unary iteration circuits (used in QROM) with improved Toffoli counts using either a constant number of clean ancillae or $n$ dirty ancillae.
    \end{itemize}
\end{enumerate}

We believe these constructions demonstrate that conditionally clean ancillae are an obviously useful tool for optimizing quantum circuits, especially in the resource-constrained regime of early fault tolerance, and that there will be many other uses beyond the specific constructions that we provide.

The rest of this paper is structured as follows. \sec{conditional_clean} formally defines conditionally clean ancillae and provides the general framework. \sec{laddered_toggle_detection} describes laddered toggle detection. \sec{mcx_constructions} presents our optimized constructions for the $n$-controlled NOT gate. \sec{produce_consume_all_and} detail the incrementer and comparator circuits.  \sec{unary_iteration} covers unary iteration. Finally, \sec{conclusion} concludes the paper.

\section{Background \& Related Work}

There has been significant prior work on optimizing quantum circuit constructions, particularly focusing on minimizing Toffoli gate counts and ancilla usage.  Many techniques rely on the use of clean or dirty ancillae to decompose multi-qubit operations into smaller gates~\cite{gidney2015blog, gidney2015blogincrement, Gidney2018halvingcostof, zhu2024unifiedarchitecturequantumlookup}.  An $n$-controlled NOT gate, denoted $\text{MCX}_n$, applies a NOT (X) gate to a target qubit controlled on the state of $n$ control qubits. Efficient constructions for MCX gates are fundamental building blocks for many quantum algorithms.

There have been several recent papers independently discovering or using conditionally clean qubits.
It seems to be an example of an idea whose time is \emph{due}.
The papers we found that described or used conditionally clean qubits are:

\begin{itemize}
    \item
    \cite{claudon2024polylogarithmicdepth}.
    Figure 1 and Section II A divides the $n$ controls into $n / b$ blocks, each of size $b$ and performs an $\text{MCX}_{b}$ with a single clean ancilla as the target to generate $b$ conditionally clean ancillae, which are then used as targets to perform $b$ different $\text{MCX}_{n/b}$ gates in parallel. Since their focus is on reducing the depth of the circuit, they obtain a decomposition with depth $\mathcal{O}(\log^3{n})$ and Toffoli count $\mathcal{O}(n\log^4{n})$ for one clean ancilla case. 

    \item
    \cite{nie2024quantum}.
    They explain the idea of conditionally clean qubits in Section 3 and note that it may be of interest to the community beyond their specific constructions. Since their focus is on getting constructions with optimal scaling, they give a construction for a $\text{MCX}_{n}$ in Figure 3 with $\mathcal{O}(n)$ Toffoli count and $\mathcal{O}(\log{n})$ Toffoli depth but do not perform constant factor analysis. We give an implementation of their construction in the supplementary material for completeness and show that their constant factors for Toffoli count is $3n$, whereas our construction has the optimal Toffoli count of $2n - 3$ while preserving the $\log{n}$-depth. 

    \item We also want to mention \cite{cryptoeprint:2024/222}, because our first realization of the concept of a conditionally clean qubit was triggered by trying to understand why Figure 2 of that paper worked.

\end{itemize}

Many of the figures in this paper are generated using the Quirk quantum circuit simulator (\href{https://algassert.com/quirk}{algassert.com/quirk}).  In Quirk diagrams:

\begin{itemize}
    \item Wires represent qubits.
    \item Gates are represented by boxes.  A filled circle on a wire represents a control qubit in the $\ket{1}$ state.  An open circle represents a control in the $\ket{0}$ state.
    \item The ``Chance" display corresponds to the green boxes and shows the probability of measuring the qubit in the $\ket{1}$ state.
    \item The ``Bloch" sphere display shows the quantum state of a single qubit on the Bloch sphere.
\end{itemize}

\begin{table}
\begin{adjustwidth}{-2cm}{-2cm}
    \centering
    \begin{tabular}{|p{2.3cm}|p{2.6cm}|p{2.7cm}|p{2.7cm}|p{1.4cm}|p{2.5cm}|p{1.4cm}|} 
 \hline
 \multicolumn{3}{|c|}{Gate} & \multicolumn{2}{c|}{Toffoli} & \multicolumn{2}{c|}{T} \\ 
  \hline
   Type & Ancilla & Source & Cost & Depth & Cost & Depth \\ 
 \hline
   $\text{MCX}_{n}$  
    & \gray One clean  & \gray \cite{gidney2015blog} & \gray $6n$  & \gray $\mathcal{O}(n)$  & \gray $24n$ & \gray $\mathcal{O}(n)$ \\ 

    & \gray One dirty  & \gray \cite{gidney2015blog} & \gray $8n$  & \gray $\mathcal{O}(n)$  & \gray $32n$ & \gray $\mathcal{O}(n)$ \\ 

    & \gray One dirty  & \gray \cite{zindorf2024efficient} & \gray \mcxoptcomplexitydirty  & \gray $\mathcal{O}(n)$  & \gray $16n - 32$ & \gray 8n - 6 \\ 

    & \gray One clean  & \gray \cite{claudon2024polylogarithmicdepth} & \gray $\mathcal{O}(n\log^4{n})$  & \gray $\mathcal{O}(\log^3{n})$  & \gray $\mathcal{O}(n\log^4{n})$ & \gray $\mathcal{O}(\log^3{n})$ \\

    & \gray One clean  & \gray \cite{nie2024quantum} & \gray $\mathcal{O}(n)$ ($\approx 3n$)  & \gray $\mathcal{O}(\log{n})$  & \gray $\mathcal{O}(n)$ & \gray $\mathcal{O}(\log{n})$ \\
   
    & One clean  &  Ours - \sec{one_clean_n_depth}   & $2n - 3$    &  $\mathcal{O}(n)$   &   $8n - 12$   &  $\mathcal{O}(n)$ \\  

    & Two clean  &  Ours - \sec{two_clean_logn_depth}   & $2n - 3$    &  $\mathcal{O}(\log{n})$   &   $8n - 12$   &  $\mathcal{O}(\log{n})$ \\  

    & One dirty  & Ours - \sec{one_dirty_n_depth}   & \mcxoptcomplexitydirty    &  $\mathcal{O}(n)$   &   $16n - 32$   &  $\mathcal{O}(n)$ \\
    & Two dirty  & Ours - \sec{two_dirty_logn_depth}   & \mcxoptcomplexitydirty    &  $\mathcal{O}(\log{n})$   &   $16n - 32$   &  $\mathcal{O}(\log{n})$ \\
 
 \hline
    $\text{Incrementer}_{n}$    
    & \gray One clean     
    & \gray \cite{gidney2015blogincrement}   & \gray $32n$   & \gray  $\mathcal{O}(n)$  & \gray  $128n$   & \gray $\mathcal{O}(n)$ \\ 
    
    & \gray One clean &
    \gray \cite{nie2024quantum}       & \gray $\mathcal{O}(n)$ ($\approx 64n$)   & \gray  $\mathcal{O}(\log^2{n})$  & \gray  $\mathcal{O}(n)$ ($\approx 256n$)   & \gray $\mathcal{O}(\log^2{n})$ \\ 
    
    & $\log_2^*{n}$ clean & Ours - \sec{n_bit_incrementer}   & $3n$    &  $\mathcal{O}(n)$   &   $12n$   &  $\mathcal{O}(n)$ \\
    
 \hline
    $U_{t \oplus (x < c)}$    
    & \gray Two clean     
    & \gray \cite{gidney2018factoring}   & \gray $\mathcal{O}(n\log{n})$   & \gray $\mathcal{O}(n\log{n})$  & \gray $\mathcal{O}(n\log{n})$   & \gray $\mathcal{O}(n\log{n})$ \\ 
    
    ($\text{LessThanC}_{n}$) & \gray One clean &
    \gray \cite{yuan2023improved}       & \gray $\mathcal{O}(n^2)$    & \gray $\mathcal{O}(n^2)$  & \gray $\mathcal{O}(n^2)$   & \gray $\mathcal{O}(n^2)$ \\

    & $\log_2^*{n}$ clean & Ours - \sec{n_bit_comparator}   & $3n$    &  $\mathcal{O}(n)$   &   $12n$   &  $\mathcal{O}(n)$ \\
  \hline
    $\text{UnaryIteration}$    
    & \gray $n=\log_2{N}$ clean     
    & \gray \cite{Babbush2018}   & \gray $N - 1$   & \gray $\mathcal{O}(N)$  & \gray $4N - 4$   & \gray $\mathcal{O}(N)$ \\   

    (balanced tree) & $\log_2^*{n}$ clean & Ours - \sec{unary_iteration_cca}   & $2.5N - 1$ &  $\mathcal{O}(N)$   &   $10N - 4$   &  $\mathcal{O}(N)$ \\

    (skew tree) & $\log_2^*{n}$ clean & Ours - \sec{unary_iteration_cca}   & $2.25N - 1$ &  $\mathcal{O}(N)$   &   $9N - 4$   &  $\mathcal{O}(N)$ \\

    (balanced tree) & $n=\log_2{N}$ dirty & Ours - \sec{unary_iteration_dirty}   & $1.5N + \mathcal{O}(n\sqrt{N})$ &  $\mathcal{O}(N)$   &   $6N + \mathcal{O}(n\sqrt{N})$   &  $\mathcal{O}(N)$ \\

    (skew tree) & $n=\log_2{N}$ dirty & Ours - \sec{unary_iteration_dirty}   & $1.25N + \mathcal{O}(n\sqrt{N})$ &  $\mathcal{O}(N)$   &   $5N + \mathcal{O}(n\sqrt{N})$   &  $\mathcal{O}(N)$ \\

 \hline
    \end{tabular}
    \caption{Comparison of prior work (shaded) to our constructions (not shaded).
    Note that $\log_2^*{n} \leq 5$ for all practical purposes.
    }
    \label{tbl:costs_table}
\end{adjustwidth}
\end{table}

\section{Conditionally Clean Ancillae} \label{sec:conditional_clean}

\begin{definition}[Conditionally Clean Ancilla] \label{definition:cca}
Let $C$ be a set of control qubits in an arbitrary state $\ket{\Psi}_c$. Let $\ket{x}_c$ be a state on the control register that satisfies a predicate $P$. 
An ancilla qubit $\ket{cca}$ is a \emph{conditionally clean ancilla} with respect to the predicate $P$ if, when $P$ is True, the ancilla qubit $\ket{cca}$ is in a known classical state $\ket{f(x)}$. 
In other words, we can write the combined state of control register $\ket{\Psi}_c$ and ancilla register $\ket{cca}_{P}$ as 

$$
    \ket{\Psi}_c\ket{cca}_{P} =  \alpha_{x}\ket{x}_c\ket{f(x)} + \alpha_{x^\perp}\ket{x^\perp}_c\ket{\psi_{f(x^\perp)}}
$$

where $\ket{x}_c$ is the control state that satisfies predicate $P$, $\alpha_{x^\perp}$ is the perpendicular control state that does not satisfy the predicate $P$, $\ket{f(x)}$ is a known classical state, $\ket{\psi_{f(x^\perp)}}$ is an arbitrary quantum state. 
\end{definition}

From the above definition, one can see that for any controlled operation $ CU_{P, S} = U_{S}\ket{x}\bra{x} + I_{S} \ket{x^\perp}\bra{x^\perp}$ (meaning $U$ is applied if and only if $P$ is true), the conditionally clean qubit $\ket{cca}$ can be used as a clean qubit to act as a temporary workspace. This is because in the subspace where $P$ is True and $U$ is applied, $\ket{cca}$ has a known classical state $\ket{f(x)}$ and thus any intermediate results stored on $\ket{cca}$ can be consumed without toggle detection.

A key example, and the one we primarily focus on in this paper, is based on the quantum implementation of the AND gate using a Toffoli gate.  A Toffoli gate has two control qubits ($c_1$, $c_2$) and a target qubit ($t$). It flips the target if and only if both controls are in the state $\ket{1}$.
Consider the following scenario:  We have two qubits, $c_1$ and $c_2$, in an unknown superposition of states, and a clean ancilla $t$ initialized to $\ket{0}$.  We apply a Toffoli gate: $\text{Toffoli}(c_1, c_2, t)$.

$$
\text{Toffoli}_{c_1, c_2, t} \ket{x}_{c_1} \ket{y}_{c_2} \ket{0}_t = \ket{x}_{c_1} \ket{y}_{c_2} \ket{x \land y}_t
$$

Now, consider the situation conditioned on the target qubit $t$ being measured in the $\ket{1}$ state (i.e. $\langle Z_t \rangle = -1$).  If $t = \ket{1}$, then we know that both $c_1$ and $c_2$ must be in the $\ket{1}$ state.  Even though $c_1$ and $c_2$ might have been in a superposition, within the subspace where $t = \ket{1}$, they behave exactly as if they were clean ancillae initialized to $\ket{1}$.  Therefore, $c_1$ and $c_2$ are conditionally clean, with the condition being $\langle Z_t \rangle = -1$, and any computation that does nothing when $\langle Z_t\rangle \neq -1$ can be compiled as if $c_1$ and $c_2$ happened to be $|1\rangle$.

A particularly interesting property of conditionally clean qubits is that they can catalyze their own production.
\fig{conditionally_clean:sfig1} shows an example of generating conditionally clean ancilla qubits when doing an $\text{AND}$ gate on a clean ancilla as a target.
\fig{conditionally_clean:sfig2} shows an example of consuming the conditionally clean qubits $a$ and $b$ to store results of intermediate computations and further generate conditionally clean ancilla qubits conditioned over multiple control bits. 

This idea of consuming a clean qubit to generate new conditionally clean qubits generalizes beyond the $\text{AND}$ gate.
For example, if you do an equality check of an $n$-bit register with a classical constant and store the result in new clean qubit, then conditioned on the clean qubit being in the $\ket{1}$ state, you know the state of the input register is equal to the classical constant.

In general, Let $f: \{0, 1\}^{n} \rightarrow \{0, 1\}^m$ be any function and let $U_f$ be a unitary that computes this function on a clean ancilla register of size $m$ such that 

$$
    \ket{x}^n \ket{0}^m \xrightarrow{U_f} \ket{x}^n \ket{f(x)}^m
$$

If there exists a pair $(x, f(x))$ such that $x$ is the only element in the domain of $f$ that maps to $f(x)$, then conditioned on the ancilla register being in the state $f(x)$, we know that the system register would be in the state $x$. Thus, we can use the system register as a conditionally clean register that is allocated in a known state $x$, use it as a temporary workspace to store intermediate computation results assuming we will consume those results conditioned on the ancilla being in state $f(x)$, and then uncompute the intermediate computation to restore the state of the system register to be in $x$. The system register in this case acts as a register of ``conditionally clean" qubits.

Thus, there are conditions which must be satisfied for us to be able to use system qubits in our computation as conditionally clean qubits:

\begin{itemize}
    \item We must consume (at least one) clean ancilla qubits to compute a function $f(x)$ that has a unique inverse $x$ such that conditioned on the ancilla register being in the state $f(x)$, we know the system register is in the state $x$ and thus can be used as conditionally clean qubits to store intermediate results of subsequent computations. 
    \item Any subsequent computation that uses the system register $x$ as conditionally clean qubits to store intermediate results must consume these results conditioned on the ancilla register being in the state $f(x)$; i.e. Let $g(y)$ be the result of the computation that uses register $x$ as a temporary workspace then we must consume $f(x) \land g(y)$ as the final output. 
\end{itemize}

In this paper, we will mostly look at utilizing the conditionally clean qubits when we wish to compute a ladder of AND gates. The function $f$ in this case is an $\text{AND}$ operation and the pair $(x=\{1\}^n, f(x)=1)$ satisfies the criteria defined above. Note that $\text{AND}$ operation also distributes nicely such that $f(x, y) = f(x) \land f(y)$ and thus we can compute $f(x)$ on a temporary qubit and use the register $x$ as a register of ``conditionally clean" qubits to store the result of computing $f(y)$, which is then consumed conditioned on the ancilla being in the state $f(x)$ such that the final consumed output is $f(x) \land f(y)$. \fig{conditionally_clean:sfig3} shows how to accumulate the $\text{AND}$ of $n$ ancillae on $2$ qubits by incrementally generating and utilizing conditionally clean ancilla qubits. 

Note that an (obvious yet important) restriction of conditionally clean ancilla qubits is that they cannot be used as temporary workspace for computations that involve the control qubits on which the conditionally clean ancilla is conditioned. For example, in \fig{conditionally_clean:sfig3} step-4 can use only $c0$ as temporary workspace to accumulate controls $c1$ and $c3$ because the remaining conditionally clean ancilla $c2c4c5$ are all conditioned on $c1$ and so cannot be used as temporary workspace for computations involving $c1$. 

\begin{figure}
\begin{adjustwidth}{-2cm}{-2cm}
\begin{subfigure}[]{0.43\linewidth}
  \centering
  \includegraphics[width=1\linewidth]{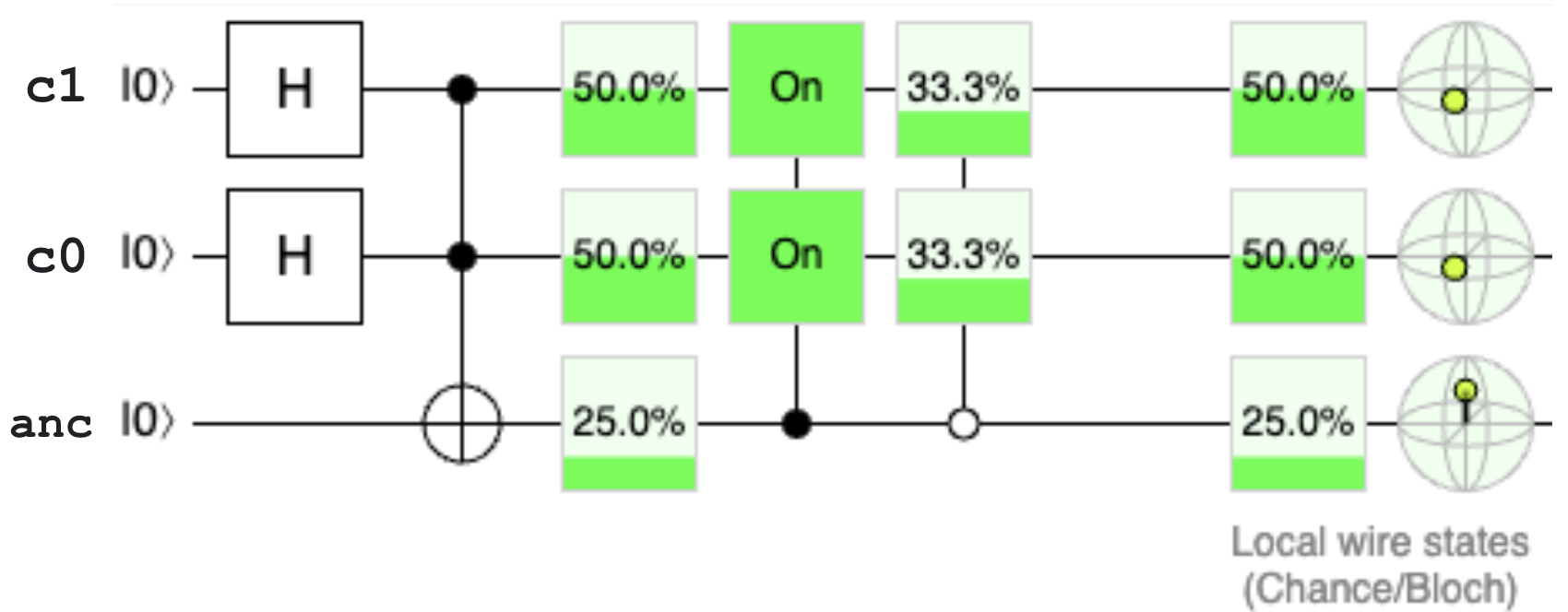}
  \caption{Application of a Toffoli gate on a clean ancilla qubit as the target generates two conditionally clean qubits. The percentage in the green boxes on each wire represent probability of that qubit being in the $\ket{1}$ state.
  Conditioned on the ancilla qubit being in the $\ket{1}$ state, the first two qubits are ``conditionally clean". These conditionally clean ancillae can be consumed from outside the subcircuit in which the qubit is conditioned from. See \fig{conditionally_clean:sfig2} for example. Here is \href{https://algassert.com/quirk\#circuit=\%7B\%22cols\%22\%3A\%5B\%5B\%22H\%22\%2C\%22H\%22\%5D\%2C\%5B\%22\%E2\%80\%A2\%22\%2C\%22\%E2\%80\%A2\%22\%2C\%22X\%22\%5D\%2C\%5B\%22Chance\%22\%2C\%22Chance\%22\%2C\%22Chance\%22\%5D\%2C\%5B\%22Chance\%22\%2C\%22Chance\%22\%2C\%22\%E2\%80\%A2\%22\%5D\%2C\%5B\%22Chance\%22\%2C\%22Chance\%22\%2C\%22\%E2\%97\%A6\%22\%5D\%5D\%7D}{quirk link} for the circuit presented above.}
  \label{fig:conditionally_clean:sfig1}
\end{subfigure}~
\begin{subfigure}[]{0.53\linewidth}
  \centering
  \includegraphics[width=1\linewidth]{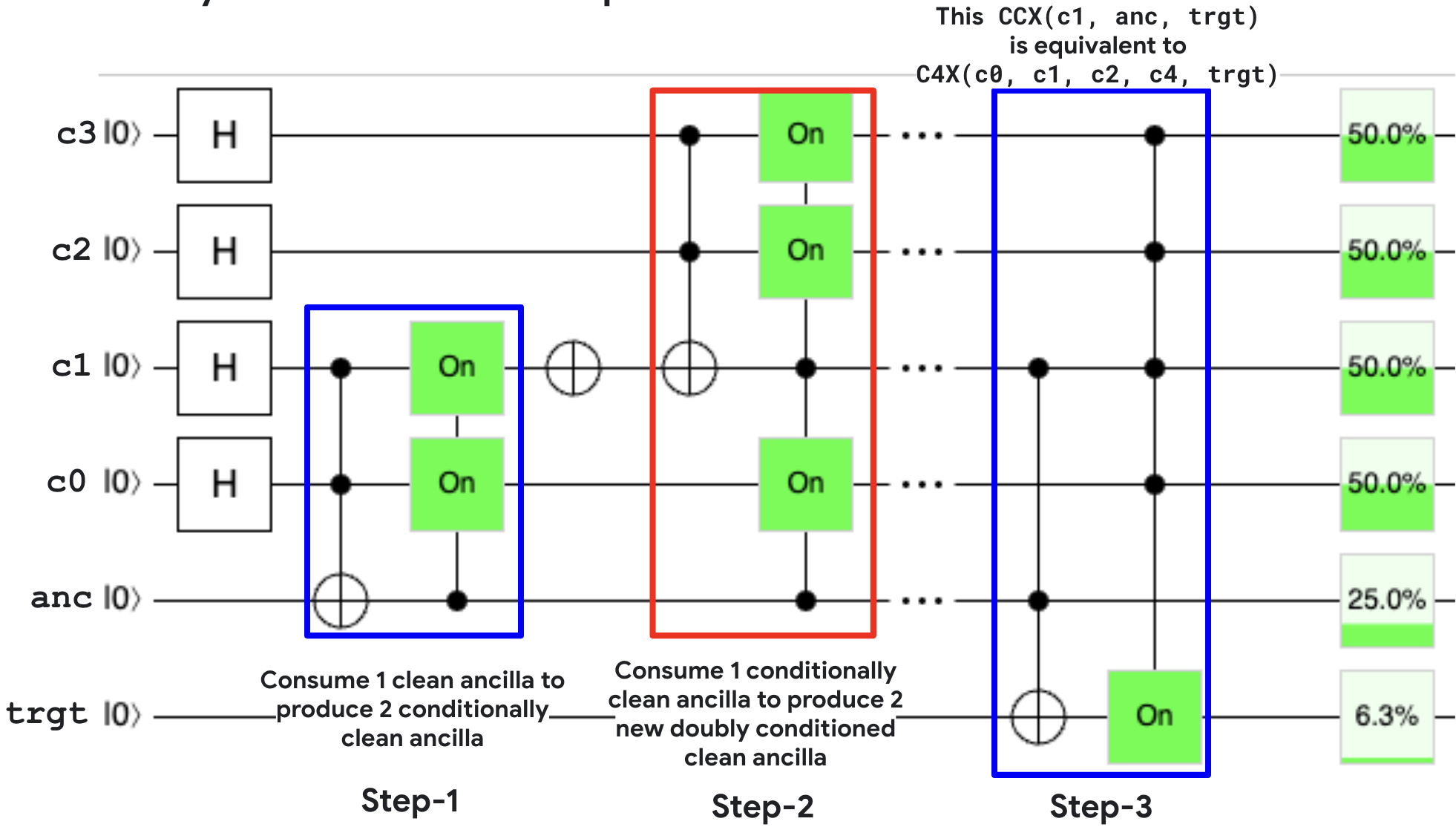}
  \caption{One of the conditionally clean ancilla from Step-1 is consumed to produce two more doubly conditioned clean ancilla in Step-2. Step-3 shows how $CCX(c_1, \text{anc}, \text{trgt})$ is now equivalent to applying $C^4X(c_0, c_1, c_2, c_3, \text{trgt})$ since $\text{anc}$ stores $c_0 \land c_1$ and conditioned on $\text{anc}$ being True, i.e. $c1=\text{True}$, $c1$ stores $c2 \land c3$}
  \label{fig:conditionally_clean:sfig2}
\end{subfigure}
\begin{subfigure}[]{\linewidth}
  \centering
  \includegraphics[width=1\linewidth]{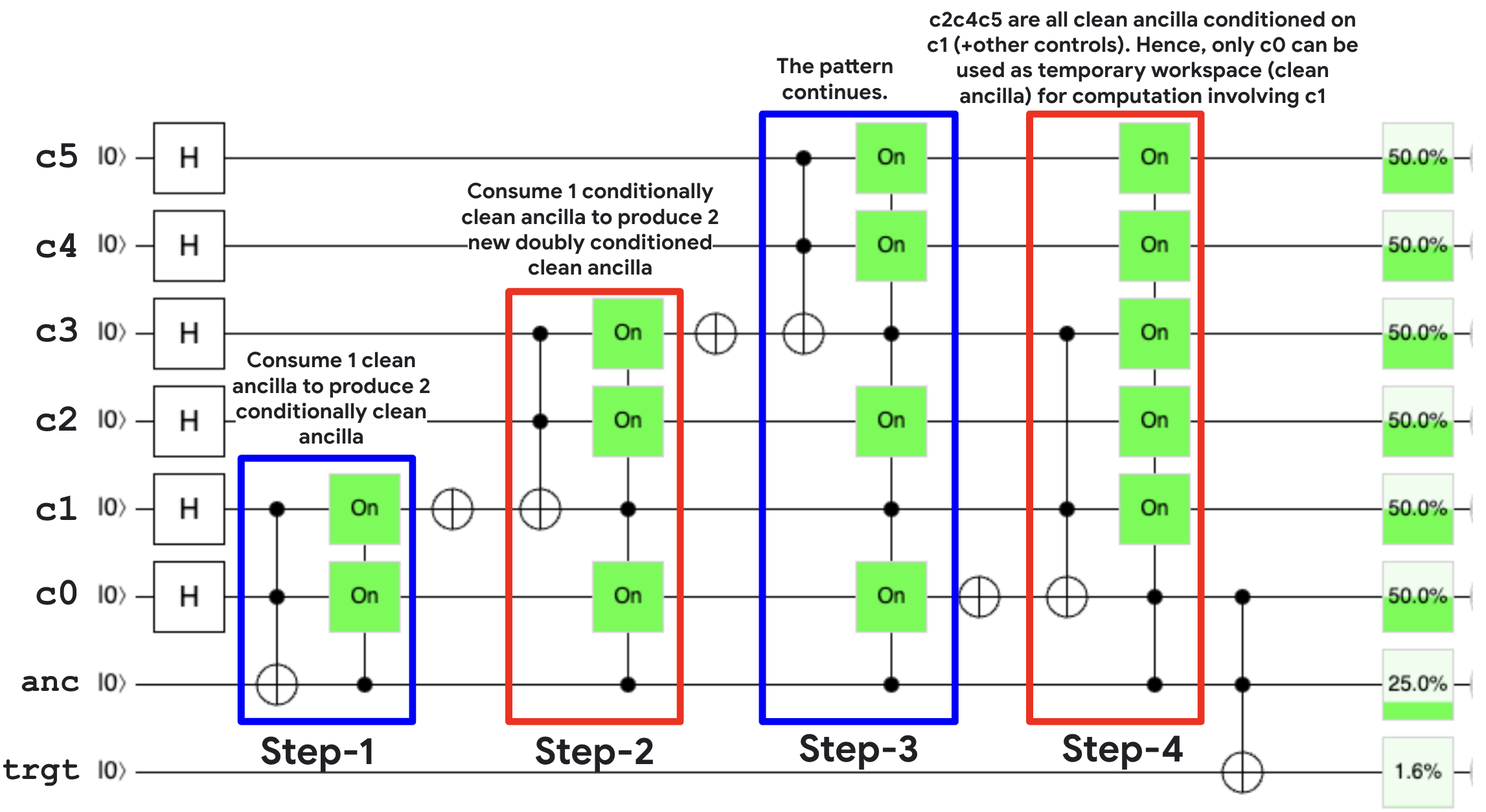}
  \caption{The staggered pattern continues so we accumulate $n$ controls on $n /2$ controls using only 1 ancilla by Step-3. Step-4 shows how conditionally clean ancilla can only be consumed in computations that do not involve the control qubits which the ancilla is conditioned on.  To accumulate the result of open controls $c_3 \land c_1$, the only $c_0$ can be used as a conditionally clean ancilla since $c_2$, $c_4$ and $c_5$ are all conditionally clean conditioned on $c_1$ and hence cannot be used as temporary workspace for computations involving $c_1$. 
  Here is a \href{https://algassert.com/quirk\#circuit=\%7B\%22cols\%22\%3A\%5B\%5B\%22H\%22\%2C\%22H\%22\%2C\%22H\%22\%2C\%22H\%22\%2C\%22H\%22\%2C\%22H\%22\%5D\%2C\%5B1\%2C1\%2C1\%2C1\%2C\%22\%E2\%80\%A2\%22\%2C\%22\%E2\%80\%A2\%22\%2C\%22X\%22\%5D\%2C\%5B1\%2C1\%2C1\%2C1\%2C\%22Chance\%22\%2C\%22Chance\%22\%2C\%22\%E2\%80\%A2\%22\%5D\%2C\%5B1\%2C1\%2C1\%2C1\%2C\%22X\%22\%5D\%2C\%5B1\%2C1\%2C\%22\%E2\%80\%A2\%22\%2C\%22\%E2\%80\%A2\%22\%2C\%22X\%22\%5D\%2C\%5B1\%2C1\%2C\%22Chance\%22\%2C\%22Chance\%22\%2C\%22\%E2\%80\%A2\%22\%2C\%22Chance\%22\%2C\%22\%E2\%80\%A2\%22\%5D\%2C\%5B1\%2C1\%2C\%22X\%22\%5D\%2C\%5B\%22\%E2\%80\%A2\%22\%2C\%22\%E2\%80\%A2\%22\%2C\%22X\%22\%5D\%2C\%5B\%22Chance\%22\%2C\%22Chance\%22\%2C\%22\%E2\%80\%A2\%22\%2C\%22Chance\%22\%2C\%22\%E2\%80\%A2\%22\%2C\%22Chance\%22\%2C\%22\%E2\%80\%A2\%22\%5D\%2C\%5B1\%2C1\%2C1\%2C1\%2C1\%2C\%22X\%22\%5D\%2C\%5B1\%2C1\%2C\%22\%E2\%80\%A2\%22\%2C1\%2C\%22\%E2\%80\%A2\%22\%2C\%22X\%22\%5D\%2C\%5B\%22Chance\%22\%2C\%22Chance\%22\%2C\%22Chance\%22\%2C\%22Chance\%22\%2C\%22Chance\%22\%2C\%22\%E2\%80\%A2\%22\%2C\%22\%E2\%80\%A2\%22\%5D\%2C\%5B1\%2C1\%2C1\%2C1\%2C1\%2C\%22\%E2\%80\%A2\%22\%2C\%22\%E2\%80\%A2\%22\%2C\%22X\%22\%5D\%5D\%7D}{quirk link} for circuit presented above.
  }
  \label{fig:conditionally_clean:sfig3}
\end{subfigure}
\caption{Observation and utilization of ``conditionally clean ancilla" qubits}
\label{fig:conditionally_clean}
\end{adjustwidth}
\end{figure}


\section{Laddered Toggle Detection}\label{sec:laddered_toggle_detection}

When implementing a controlled version of a self inverse operation $U$, one can replace the use of a clean ancilla qubit with a dirty ancilla qubit by repeating the self inverse controlled operation twice. If the control is True and dirty ancilla flips, the self inverse $C-U$ is applied on both the branches $\ket{0}$ and $\ket{1}$ of the dirty ancilla. If the control is False, $C-U$ is applied twice on the $\ket{1}$ branch and zero times on the $\ket{0}$ branch, effectively applying an identity operation. This well known trick of substituting a dirty ancilla in the place of a clean ancilla is commonly~\cite{gidney2015blog, nie2024quantum} called ``toggle detection" and we show how it works in \fig{toggle_detection:sfig1}.

If the implementation of $C-U$ further borrows a clean ancilla qubit which we wish to replace with a dirty ancilla qubit, then a naive strategy would be to recursively repeat the toggle detection strategy described above. Every time we apply toggle detection, we incur a 2x overhead since we need to apply $C-U$ twice. Thus, if we wish to replace $n$ clean ancillae with $n$ dirty ancillae, we can end up incurring an overhead that scales exponentially as $2^n$.

We introduce Laddered Toggle Detection, a technique that avoids this exponential overhead.  The key observation is that if the subsequently borrowed dirty ancillae are different from the control qubits used in the initial toggle detection, these borrowed ancillae can be treated as if they were clean. \fig{toggle_detection:sfig2} demonstrates this principle and \fig{toggle_detection:sfig3} shows an application of this technique to derive a new construction for decomposing $n$-bit Toffoli into $4n-8$ Toffoli using $n - 2$ borrowed qubits.

We note that all our circuit constructions presented below for the clean ancilla case can be updated to use dirty ancilla with a 2x Toffoli overhead using Laddered Toggle Detection. 

\begin{figure}
\begin{adjustwidth}{-2cm}{-2cm}
\begin{subfigure}{\linewidth}
  \centering
  \includegraphics[width=\linewidth]{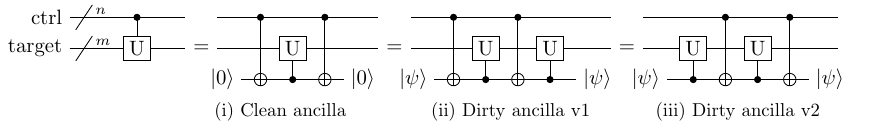}
  \caption{Toggle detection to substitute a single clean ancilla with a dirty ancilla.}
  \label{fig:toggle_detection:sfig1}
\end{subfigure}
\par\bigskip
\begin{subfigure}{\linewidth}
  \centering
  \includegraphics[width=\linewidth]{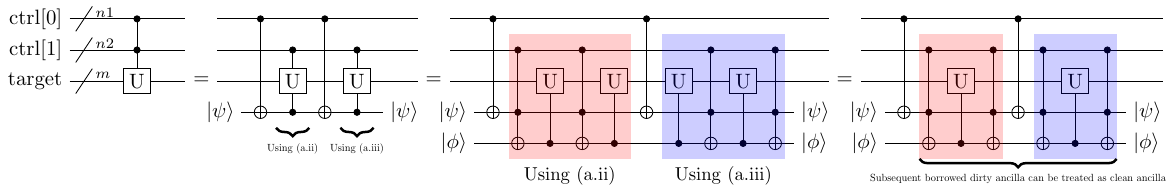}
  \caption{Subsequent dirty qubits can be treated as clean qubits since you get toggle detection for free. Note that in addition to $U^2=I$, another constraint for this to work is that the decomposition of $CU$ should not use the ctrl[0] qubits from top as dirty ancillae.}
  \label{fig:toggle_detection:sfig2}
\end{subfigure}
\par\bigskip
\begin{subfigure}{\linewidth}
  \centering
  \includegraphics[width=\linewidth]{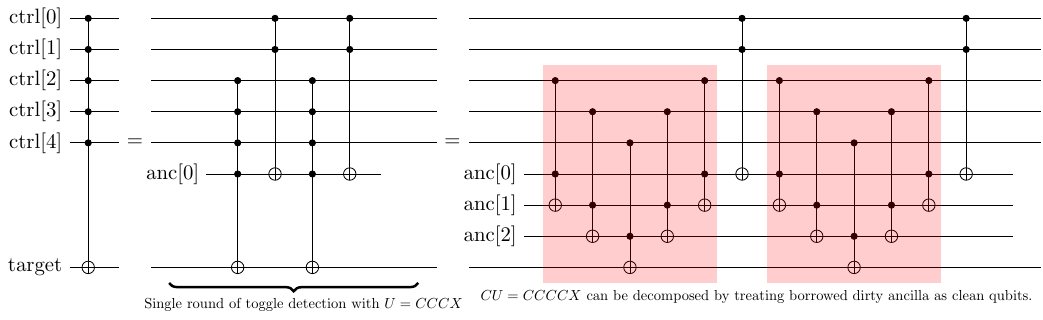}
  \caption{An example of utilizing the laddered toggle detection trick to decompose $n$-bit Toffoli into $4n - 8$ Toffoli using $n - 2$ borrowed ancilla qubits. Note that this circuit has the same cost as the ``arrow pointing away from the target" construction in \cite{gidney2015blog} but utilizes the toggle detection trick so the ``arrow is pointing towards the target" now.
  Here is \href{https://algassert.com/quirk\#circuit=\%7B\%22cols\%22\%3A\%5B\%5B\%22H\%22\%2C\%22H\%22\%2C\%22H\%22\%2C\%22H\%22\%2C\%22H\%22\%2C\%22Counting3\%22\%5D\%2C\%5B\%22\%E2\%80\%A6\%22\%2C\%22\%E2\%80\%A6\%22\%2C\%22\%E2\%80\%A6\%22\%2C\%22\%E2\%80\%A6\%22\%2C\%22\%E2\%80\%A6\%22\%2C\%22\%E2\%80\%A6\%22\%2C\%22\%E2\%80\%A6\%22\%5D\%2C\%5B1\%2C1\%2C\%22\%E2\%80\%A2\%22\%2C1\%2C1\%2C\%22\%E2\%80\%A2\%22\%2C\%22X\%22\%5D\%2C\%5B1\%2C1\%2C1\%2C\%22\%E2\%80\%A2\%22\%2C1\%2C1\%2C\%22\%E2\%80\%A2\%22\%2C\%22X\%22\%5D\%2C\%5B1\%2C1\%2C1\%2C1\%2C\%22\%E2\%80\%A2\%22\%2C1\%2C1\%2C\%22\%E2\%80\%A2\%22\%2C\%22X\%22\%5D\%2C\%5B1\%2C1\%2C1\%2C\%22\%E2\%80\%A2\%22\%2C1\%2C1\%2C\%22\%E2\%80\%A2\%22\%2C\%22X\%22\%5D\%2C\%5B1\%2C1\%2C\%22\%E2\%80\%A2\%22\%2C1\%2C1\%2C\%22\%E2\%80\%A2\%22\%2C\%22X\%22\%5D\%2C\%5B\%22\%E2\%80\%A2\%22\%2C\%22\%E2\%80\%A2\%22\%2C1\%2C1\%2C1\%2C\%22X\%22\%5D\%2C\%5B1\%2C1\%2C\%22\%E2\%80\%A2\%22\%2C1\%2C1\%2C\%22\%E2\%80\%A2\%22\%2C\%22X\%22\%5D\%2C\%5B1\%2C1\%2C1\%2C\%22\%E2\%80\%A2\%22\%2C1\%2C1\%2C\%22\%E2\%80\%A2\%22\%2C\%22X\%22\%5D\%2C\%5B1\%2C1\%2C1\%2C1\%2C\%22\%E2\%80\%A2\%22\%2C1\%2C1\%2C\%22\%E2\%80\%A2\%22\%2C\%22X\%22\%5D\%2C\%5B1\%2C1\%2C1\%2C\%22\%E2\%80\%A2\%22\%2C1\%2C1\%2C\%22\%E2\%80\%A2\%22\%2C\%22X\%22\%5D\%2C\%5B1\%2C1\%2C\%22\%E2\%80\%A2\%22\%2C1\%2C1\%2C\%22\%E2\%80\%A2\%22\%2C\%22X\%22\%5D\%2C\%5B\%22\%E2\%80\%A2\%22\%2C\%22\%E2\%80\%A2\%22\%2C1\%2C1\%2C1\%2C\%22X\%22\%5D\%2C\%5B\%22\%E2\%80\%A6\%22\%2C\%22\%E2\%80\%A6\%22\%2C\%22\%E2\%80\%A6\%22\%2C\%22\%E2\%80\%A6\%22\%2C\%22\%E2\%80\%A6\%22\%2C\%22\%E2\%80\%A6\%22\%2C\%22\%E2\%80\%A6\%22\%5D\%2C\%5B\%22\%E2\%80\%A2\%22\%2C\%22\%E2\%80\%A2\%22\%2C\%22\%E2\%80\%A2\%22\%2C\%22\%E2\%80\%A2\%22\%2C\%22\%E2\%80\%A2\%22\%2C1\%2C1\%2C1\%2C\%22X\%22\%5D\%2C\%5B\%22H\%22\%2C\%22H\%22\%2C\%22H\%22\%2C\%22H\%22\%2C\%22H\%22\%2C\%22Uncounting3\%22\%5D\%5D\%7D}{quirk link} for the construction presented above.
  }
  \label{fig:toggle_detection:sfig3}
\end{subfigure}
\caption{Laddered Toggle detection can be used to substitute multiple clean ancillae with dirty ancillae with a single round of toggle detection. $U^2 = I$}
\label{fig:toggle_detection}
\end{adjustwidth}
\end{figure}

\section{Application of conditionally clean ancilla to \texorpdfstring{$n$}{n}-bit Toffoli circuits}\label{sec:mcx_constructions}

In this section, we will present a number of new circuit constructions for decomposing $n$-bit Toffoli into Toffoli gates using conditionally clean qubits. We first reduce the problem of constructing circuit decomposition assisted by conditionally clean ancilla into the following abstract computer science problem and show that each strategy for solving the question listed below can lead to a circuit construction for decomposing $n$-bit Toffoli gates. 

\subsection{An abstract problem for \texorpdfstring{$n$}{n}-bit Toffoli Decomposition}

\newtheorem{problem}{Problem}
\begin{problem}[Ancilla State Management] \label{problem:problem_1}
Given an array $A$ with $n+1$ elements, such that $A = [1, 0, 0, ..., 0]$ i.e. initially $A[0]= 1$ and $A[i] = 0$ for $0 < i \leq n$. In each step, we can perform the following operation:
\begin{itemize}
    \item Choose indices $t$, $x$, $y$ such that $t < x < y$ and $A[t] = 1$, $A[x] = A[y] = 0$
    \item Flip the values of $A[x], A[y], A[t]$; i.e. set $A[t] = 0$, $A[x] = A[y] = 1$
\end{itemize}
Our objective is to perform a sequence of operations to minimize the number of unmarked elements ($0$'s) in the array. We can characterize any valid scheme of performing the operations using the following 3 parameters -
\begin{itemize}
    \item $K$: the number of unmarked elements (i.e. $i$ such that $A[i] = 0$) at the end of the procedure
    \item $T$: the number of operations used
    \item $D$: the depth of the sequence of chosen operations, where two operations can be performed in parallel if their $(x, y, t)$ tuples are disjoint
\end{itemize}
\end{problem}

\begin{theorem}
Each solution to \problemref{problem_1} defined above can be mapped to a circuit decomposition for accumulating the AND of $n$-qubits into $K$ qubits using exactly $T$ Toffoli gates, Toffoli depth of $D$ and $1$ clean ancilla qubit
\end{theorem}
\begin{proof}
    We can map a sequence of operations satisfying the constraints to a circuit decomposition as follows:
    \begin{itemize}
        \item For $n$-bit Toffoli, each of the $n$ system qubits can be mapped to indices $i=1..n$ in the array $A$.
        \item The index $i=0$, which is initially marked, corresponds to the 1 clean ancilla qubit required by the decomposition.
        \item During the course of the procedure, each index $i$ such that $A[i] = 1$ corresponds to a conditionally clean ancilla qubit which can be consumed as a resource to accumulate the AND of qubits on the right of it.
        \item Each operation on a tuple $(x, y, t)$ can be replaced with a gate sequence - [$\text{Toffoli}(x, y, t), \text{X}(t)$] if $t$ is a conditionally clean qubit (i.e. $t > 0$). If $t$ is a clean qubit (i.e. $t=0$), then we can use an $\text{AND}(x, y, t)$ gate instead. Thus, each operation computes $x \land y$ and saves the result on the conditionally clean qubit $t$. Thus, after the operation is performed, $x$, $y$ are available as conditionally clean qubits and $t$ is flipped from a conditionally clean qubit to a system qubit. 
        \item The constraint that for each operation $t < x < y$ ensures that at any point in time, a conditionally clean qubit $t$ is clean conditioned on a set of control qubits to the left of it and is used as temporary workspace to store results of computations for qubits on the right of it. This invariant ensures that property for conditional cleanliness is always satisfied. 
    \end{itemize}
    
\end{proof}

\subsection{\texorpdfstring{$n$}{n}-bit Toffoli into \texorpdfstring{\mcxoptcomplexityclean}{2n-3} Toffoli and \texorpdfstring{$O(n)$}{O(n)} depth using 1 clean ancilla}\label{sec:one_clean_n_depth}

One solution to \problemref{problem_1} can be obtained via a greedy strategy where in each step, you pick the rightmost marked index $t$ such that $A[t]=1$ and there are at least two indices $x$ and $y$ such that $t < x < y$ and $A[x]=A[y]=0$. Pick the leftmost such pair of $x$ and $y$ and apply the operation on the tuple $(x, y, t)$. 

This greedy procedure gives us a solution with $K=2, T=n-2, D=n-2$ and we can map it to a circuit for decomposing $n$-qubit Toffoli into \mcxoptcomplexityclean Toffoli and $O(n)$ depth using 1 clean ancilla as shown in \fig{mcx_one_clean_n_depth}

\begin{figure}
\begin{adjustwidth}{-2cm}{-2cm}
\begin{subfigure}{\linewidth}
\includegraphics[width=1\linewidth]{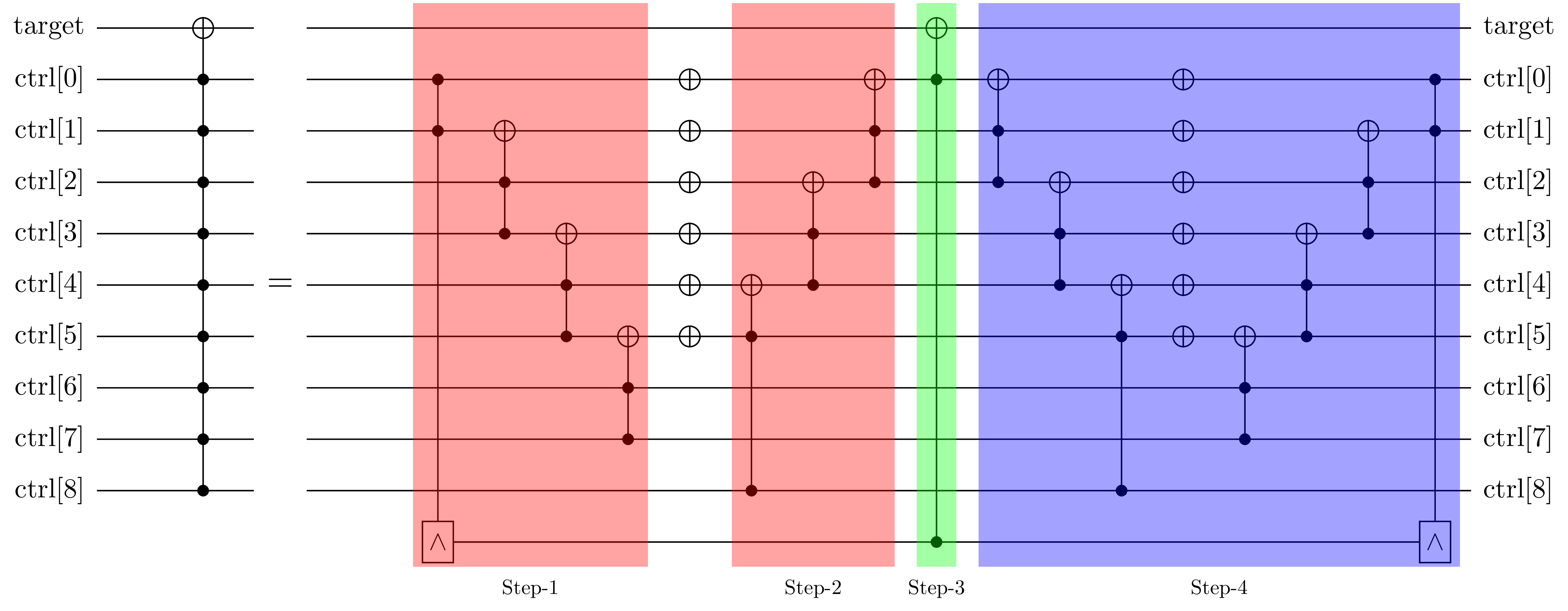}
\caption{Step-1 is a ``up" Toffoli ladder that uses $\lfloor \frac{n}{2} \rfloor$ Toffolis to accumulate the $n$ control values on $\lfloor \frac{n}{2} \rfloor$ controls using conditionally clean ancillae.  Step-2 then uses a ``down" Toffoli ladder with $n - \lfloor n / 2 \rfloor - 2$ Toffolis to accumulate the AND of all controls on conditionally clean ancilla ctrl[0]. Step-3 then uses 1 Toffoli, which doesn't appear in a compute/uncompute pair, to apply the n-qubit $C^nX$ gate on the target qubit. Step-4 uncomputes the ``up" and ``down" Toffoli ladders to return the clean ancilla and intermediate conditional ancilla to their original state. Note that the uncomputation takes 1 less Toffoli since Toffoli on clean ancilla (AND gate) can be uncomputed using measurement based uncomputation and only clifford gates. \cite{Gidney2018halvingcostof}
}
\label{fig:mcx_one_clean_n_depth:sfig1}
\end{subfigure}
\begin{subfigure}{\linewidth}
\includegraphics[width=1\linewidth]{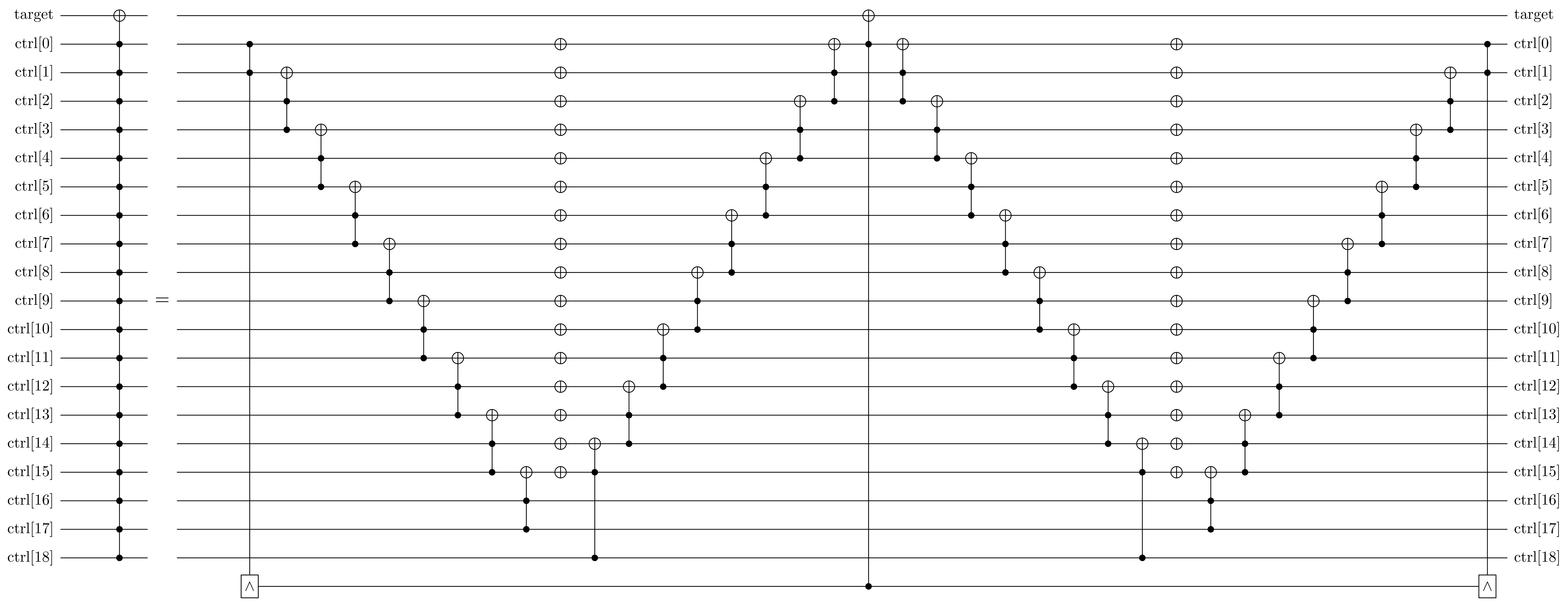}
\caption{Example circuit for decomposing a 19-bit Toffoli into 33 Toffoli + 1 $\text{And}/\text{And}^{\dagger}$ pair using 1 clean ancilla.}
\label{fig:mcx_one_clean_n_depth:sfig2}
\end{subfigure}
\caption{Decomposition of $n$-qubit Toffoli into \mcxoptcomplexityclean Toffoli using 1 clean ancilla.}
\label{fig:mcx_one_clean_n_depth}
\end{adjustwidth}
\end{figure}

\subsection{\texorpdfstring{$n$}{n}-bit Toffoli into \texorpdfstring{\mcxoptcomplexityclean}{2n - 3} Toffoli and \texorpdfstring{$\mathcal{O}(\log{n})$}{O(log(n))} depth using 2 clean ancilla} \label{sec:two_clean_logn_depth}

Another solution to \problemref{problem_1} that minimizes depth can be described as follows:
\begin{itemize}
    \item At the $i$'th timestep, we maintain the invariant that $A$ has leftmost $i$ elements in the $0$ state, the next $2^{i}$ elements in the $1$ state and all the remaining elements in the 0 state. Thus, the procedure ends in $\log{n}$ such timesteps and $K=\log{n}$ unmarked elements remain in the array at the end of the procedure. 
    For example, over the course of the algorithm, the array $A$ for $n=36$ would look like:
        \begin{itemize}
            \item At $i=0$, $A = 1000000000000000000000000000000000000$
            \item At $i=1$, $A = 0110000000000000000000000000000000000$
            \item At $i=2$, $A = 0011110000000000000000000000000000000$
            \item At $i=3$, $A = 0001111111100000000000000000000000000$
            \item At $i=4$, $A = 0000111111111111111100000000000000000$
            \item At $i=5$, $A = 0000011111111111111111111111111111111$
        \end{itemize}
    \item In the $i$'th timestep, we utilize the leftmost one of the $2^{i}$ marked elements as the target and remaining $2^{i} - 1$ marked elements as a temporary workspace to flip the next $2^{i} + 1$ elements using $2^{i}$ operations.  
    Thus, we end up with $(2^{i} - 1) + (2^{i} + 1) = 2^{i+1}$ marked elements total and the desired invariant is satisfied. 
    The operations required to achieve this transformation can be described as a sequence of $i+1$ steps as follows:
    \begin{itemize}
        \item         
        Let $P_{j}$ and $Q_{j}$ be the set of indices that correspond to the available temporary workspace and the set of unmarked elements in the $j$'th step, where $j = 0, 1, \dots i$. 
        Thus, $P_{0} = \{i, i + 1, \dots i + 2^{i} - 1\}$ and $Q_{0} = \{i + 2^{i}, i + 2^{i} + 1, \dots i + 2^{i + 1} \}$ and we want $P_{i+1} = \{\}$ and $Q_{i+1} = \{i\}$.
        \item In the $j$'th step, let $P_{j} = \{p_{0}, p_{1}, \dots, p_{2^{i - j} - 1}\}$ and $Q_{j} = \{q_0, q_1, \dots, q_{2^{i - j}}\}$. 
        To transition to $P_{j+1}$ and $Q_{j+1}$, for $0 \leq j < i$, we apply $2^{i-j-1}$ operations of the form $(t_{k}, x_{k}, y_{k})$ where $t_{k} = p_{2^{i - j} - 1 - k}, x_{k}=q_{2k+1}, y_{k}=q_{2k+2}$ and update $Q_{j + 1} = Q_{j} \setminus \{x_{k}, y_{k}\} \cup \{ t_{k} \}$ and $P_{j + 1} = P_{j} \setminus \{t_{k}\}$ where $k = 0, 1, \dots, 2^{i-j-1} - 1$.
        For $j = i$, the transition simply involves 1 operation of the form $(p_0, q_0, q_1)$.
        \item Thus, a total of $1 + \sum_{j=0}^{i-1}2^{i-j-1} = 2^{i}$ operations are applied in the $i$'th step. 
    \end{itemize}
\end{itemize}

This procedure gives us a solution with $K=\log{n}$, $T=n-\log{n}$, $D=\log{n}$. In order to map it to a circuit for decomposing $n$-qubit Toffoli into \mcxoptcomplexityclean Toffoli and $O(\log{n})$ depth, we can recursively invoke the linear depth procedure from \sec{one_clean_n_depth} and thus obtain a $O(\log{n})$ depth decomposition using 2 clean ancilla as shown in \fig{mcx_two_clean_logn_depth}

\begin{figure}
\begin{adjustwidth}{-2cm}{-2cm}
\begin{subfigure}{0.4\linewidth}
\includegraphics[width=\linewidth]{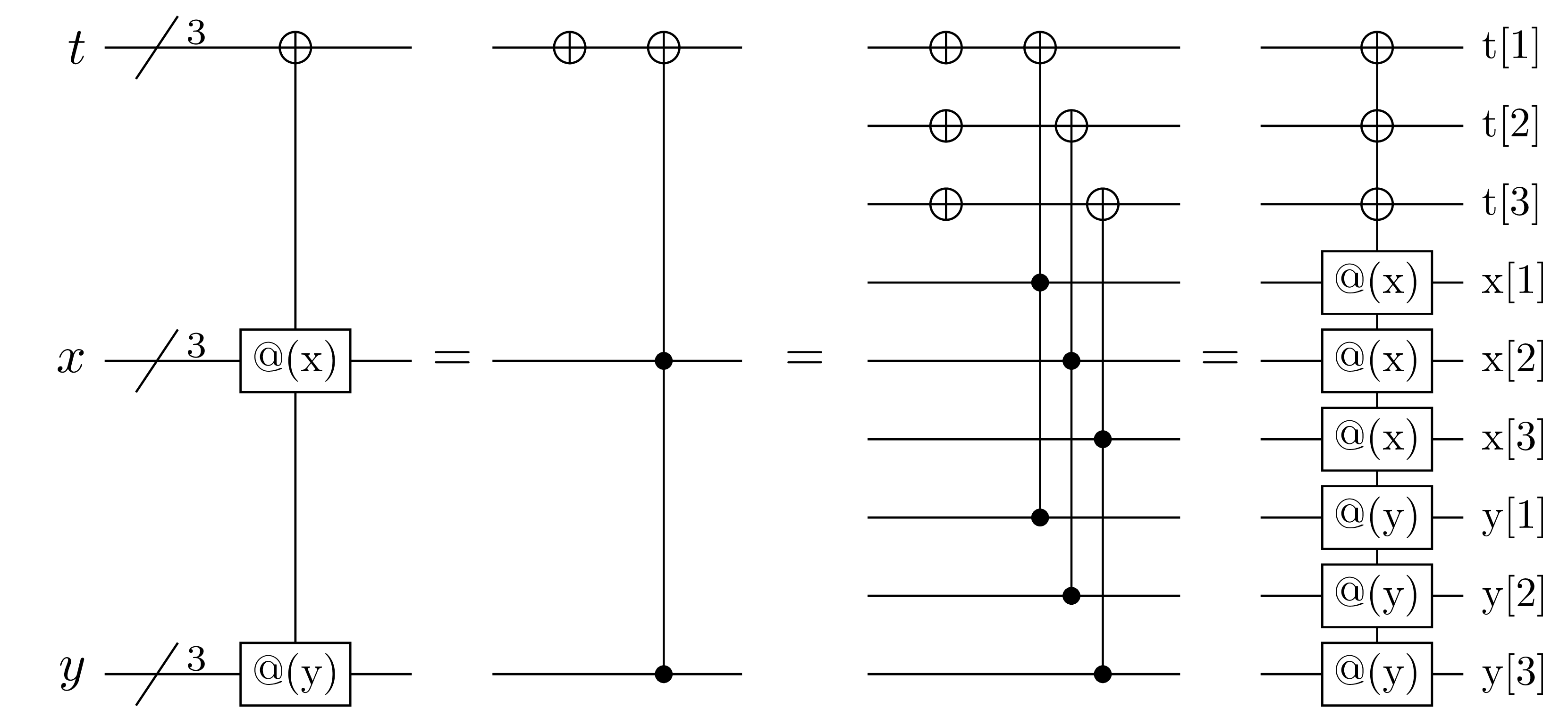}
\caption{A circuit primitive used in our construction which first toggles the $n$ conditionally clean ancilla qubits and then applies $n$ Toffoli gates in parallel. The depth of this primitive is 2.}
\label{fig:mcx_two_clean_logn_depth:sfig1}
\end{subfigure}~
\begin{subfigure}{0.58\linewidth}
\includegraphics[width=\linewidth]{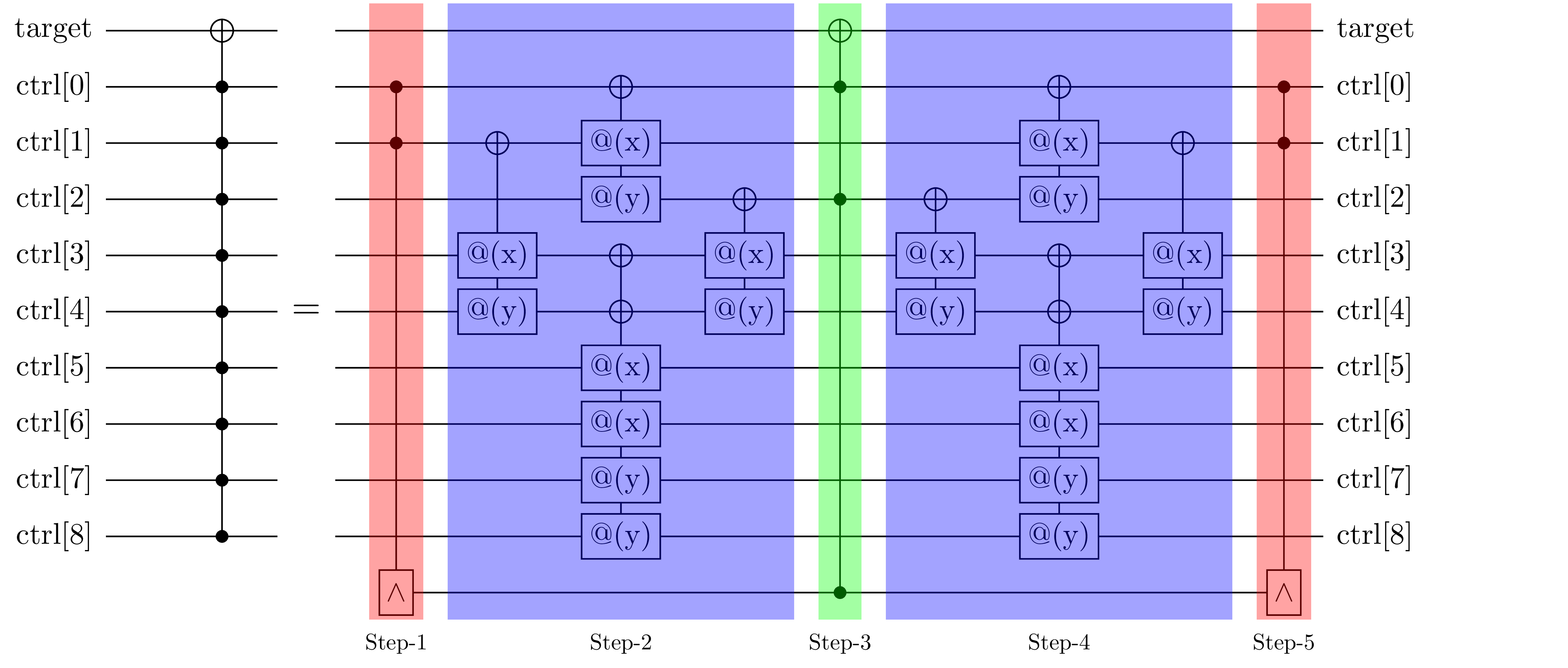}
\caption{Step-1 uses the one clean ancilla to generate two conditionally clean ancillae. Step-2 uses a log-depth ladder of $n - \log{n} - 2$ Toffolis to accumulate the $\text{AND}$ of the remaining $n - 2$ controls on $\log{n}$ qubits. Step-3 uses a $(\log{n} + 1)$-bit Toffoli to accumulate the $\text{AND}$ of all $n$ bits on the target qubit. Step-4 and 5 uncompute Step-2 and 1 respectively. Step-3 can be implemented by recursively invoking a linear depth construction using 1 additional clean ancilla from \sec{one_clean_n_depth}}
\label{fig:mcx_two_clean_logn_depth:sfig2}
\end{subfigure}
\begin{subfigure}{\linewidth}
\includegraphics[width=\linewidth, center]{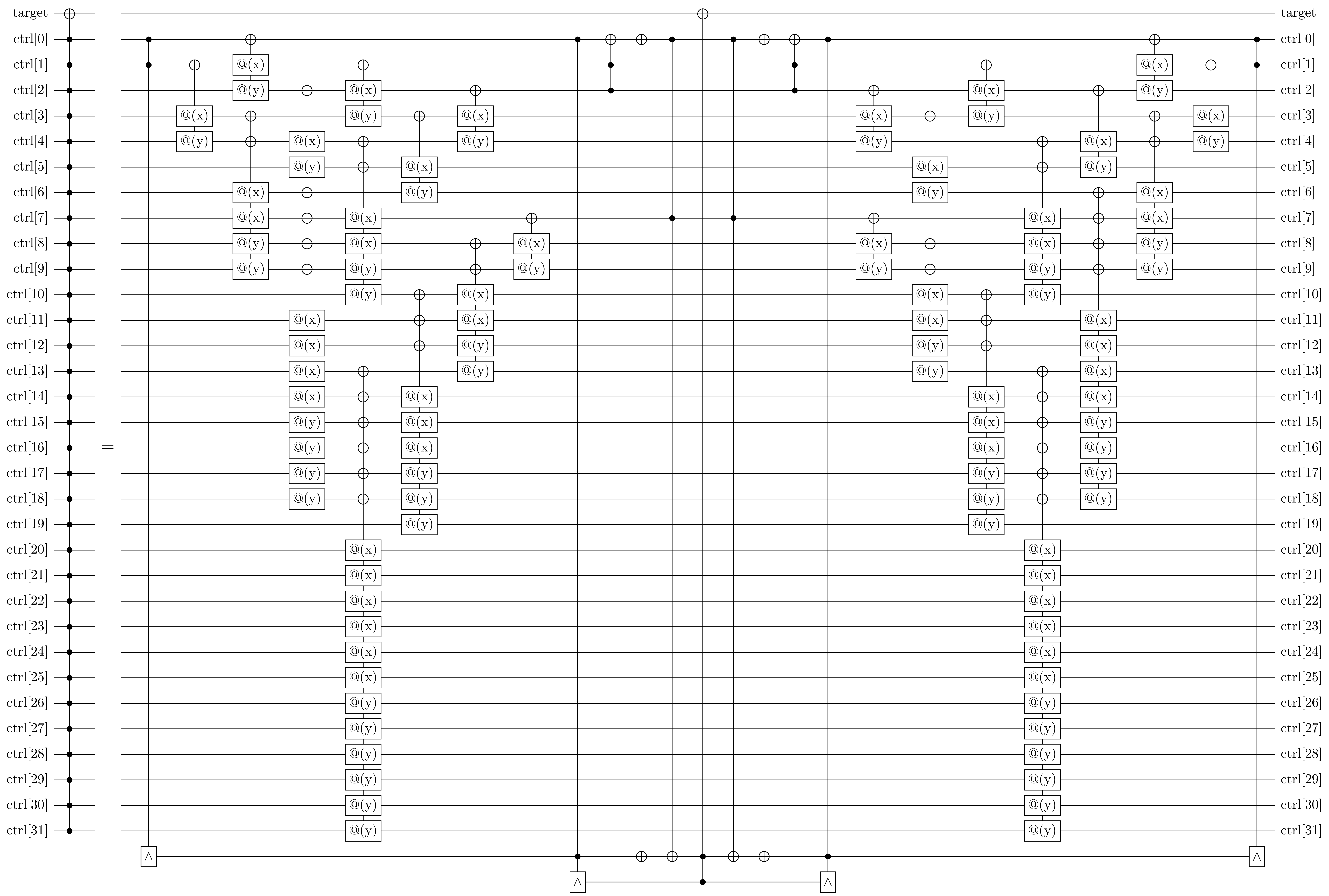}
\caption{Decomposition of $32$-bit Toffoli into $61$ Toffolis and $\mathcal{O}(\log_2{n})$ depth using 2 clean ancilla.}
\label{fig:mcx_two_clean_logn_depth:sfig3}
\end{subfigure}
\caption{Decomposition of $n$-qubit Toffoli into $2n - 3$ Toffoli and $O(\log{n})$ depth using $2$ clean ancilla. Note that the $2$ Toffolis acting on clean ancilla qubits are replaced with $2$ $\text{AND} / \text{AND}^{\dagger}$ pairs.}
\label{fig:mcx_two_clean_logn_depth}
\end{adjustwidth}
\end{figure}

\subsection{\texorpdfstring{$n$}{n}-bit Toffoli into \texorpdfstring{\mcxoptcomplexitydirty}{4n-8} Toffoli and \texorpdfstring{$O(n)$}{O(n)} depth using 1 dirty ancilla} \label{sec:one_dirty_n_depth}

This construction can be understood as a modification of the strategy described in \sec{one_clean_n_depth}, adapted to use a dirty ancilla instead of a clean one. We still follow the same general approach of accumulating ANDs using a Toffoli ladder, but we leverage the toggle detection trick (described in \sec{laddered_toggle_detection}) to handle the unknown initial state of the dirty ancilla. The key observation is that the ancilla qubit used during the clean ancilla decomposition described in \sec{one_clean_n_depth} only becomes the target of a single $CCX$ gate and stores $c_0 \land c_1$ for the first 2 controls. Thus, we can use the toggle detection trick from \sec{laddered_toggle_detection} to repeat the Toffoli ``up" and ``down" ladders twice to get a construction where the ancilla qubit can be a dirty ancilla qubit. This is shown in \fig{one_dirty_n_depth}.

\begin{figure}
\begin{adjustwidth}{-2cm}{-2cm}
\begin{subfigure}{\linewidth}
\includegraphics[width=\linewidth]{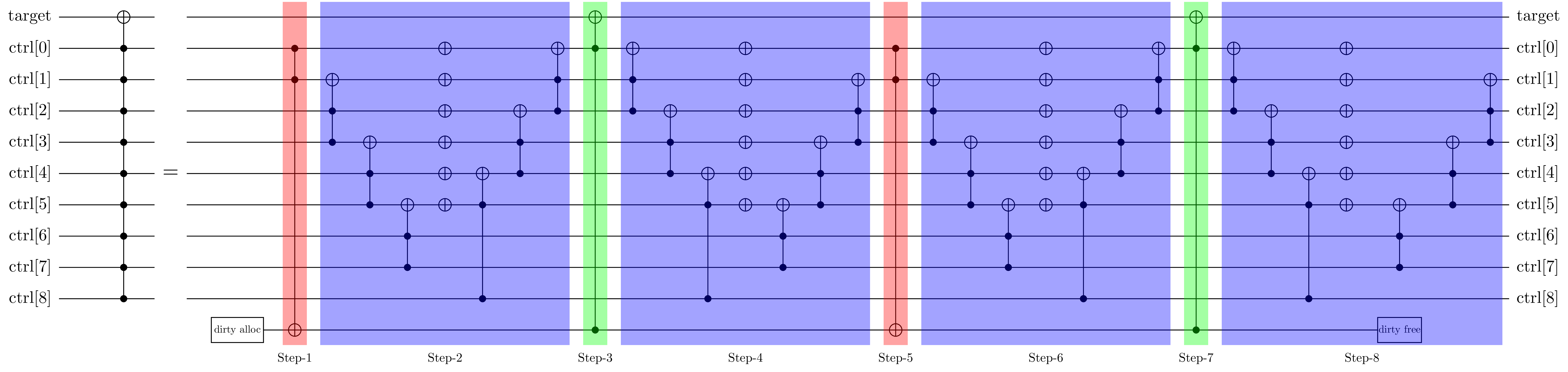}
\caption{Step 1 and 5 are applications of $CCX$ on dirty ancilla as target for toggle detection. The ``up" and ``down" Toffoli ladders in between at steps 2, 4, 6, 8 are responsible to compute/uncompute the AND of all remaining controls on conditionally clean ctrl[0]. Step-3 and 7 flip the target when all accumulated control bits are ON. When all controls are ON, the target is flipped once in Step-3 if the dirty ancilla is initially OFF and once in Step-7 if the dirty ancilla is initially ON. Each of step 2, 4, 6 and 8 requires $n - 3$ Toffoli so the construction requires a total of \mcxoptcomplexitydirty Toffoli.
}
\label{fig2:sfig1}
\end{subfigure}
\begin{subfigure}{\linewidth}
    \resizebox{\linewidth}{!}{
    \includegraphics{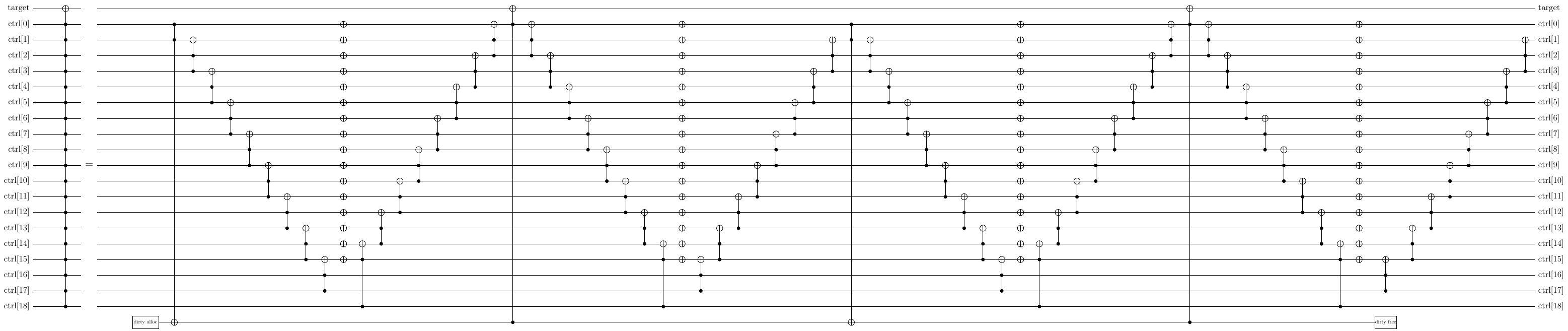}
    }
    \caption{Example circuit for decomposing a 19-bit Toffoli into 68 Toffoli using 1 dirty ancilla.}
\label{fig:one_dirty_n_depth:sfig2}
\end{subfigure}
\caption{Decomposition of $n$-qubit Toffoli into \mcxoptcomplexitydirty Toffoli using and $O(n)$ depth using 1 dirty ancilla.}
\label{fig:one_dirty_n_depth}
\end{adjustwidth}
\end{figure}

\subsection{\texorpdfstring{$n$}{n}-bit Toffoli into \texorpdfstring{\mcxoptcomplexitydirty}{4n-8} Toffoli and \texorpdfstring{$O(\log{n})$}{O(log(n))} depth using 2 dirty ancilla}\label{sec:two_dirty_logn_depth}

This construction builds upon the log-depth construction with two clean ancillae (\sec{two_clean_logn_depth}), but replaces the clean ancillae with dirty ones. For the first round of decomposition from \sec{two_clean_logn_depth}, we use standard toggle detection (also described in \sec{laddered_toggle_detection}) to replace the clean qubit with a dirty qubit. The recursive decomposition using the linear depth procedure and a second clean ancilla can treat a borrowed dirty ancilla as a clean ancilla, using laddered toggle detection from \sec{laddered_toggle_detection}, and thus the recursive decomposition remains identical to the clean ancilla case. This is shown in \fig{two_dirty_logn_depth}

\begin{figure}
\begin{adjustwidth}{-2cm}{-2cm}
\begin{subfigure}{\linewidth}
\includegraphics[width=\linewidth]{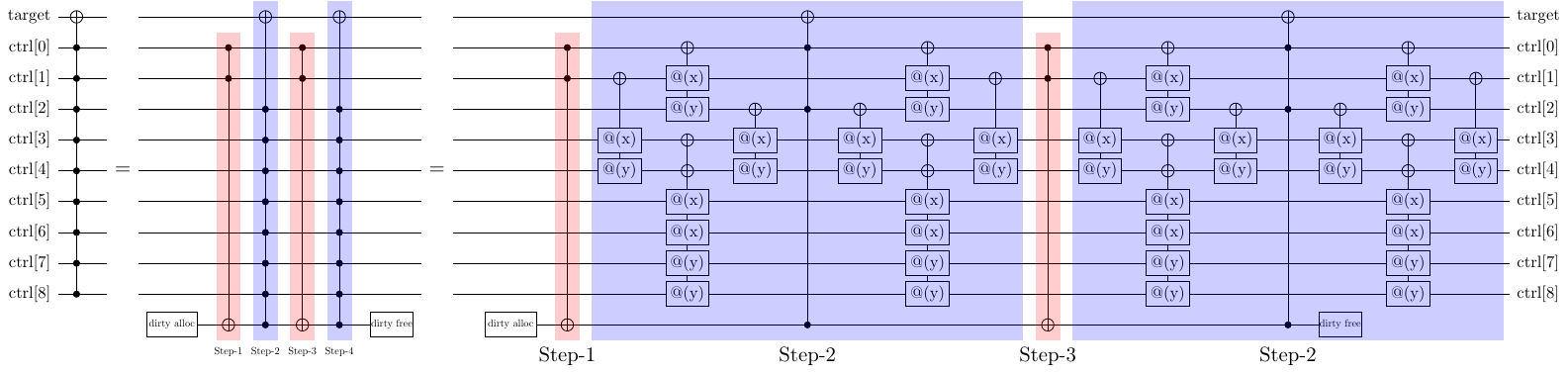}
\caption{Step-1 and 3 implement toggle detection using the first borrowed ancilla. Steps-2 and 4 use the laddered toggle detection trick from \fig{toggle_detection} to decompose an $(n - 1)$-bit Toffoli into $2n - 5$ Toffoli using 1 dirty ancilla, by treating the dirty ancilla as clean and using decomposition from \fig{mcx_two_clean_logn_depth}. The overall Toffoli complexity is $4n - 8$}
\label{fig:two_dirty_logn_depth:sfig1}
\end{subfigure}
\begin{subfigure}{\linewidth}
    \noindent
    \includegraphics[width=0.999\linewidth, left]{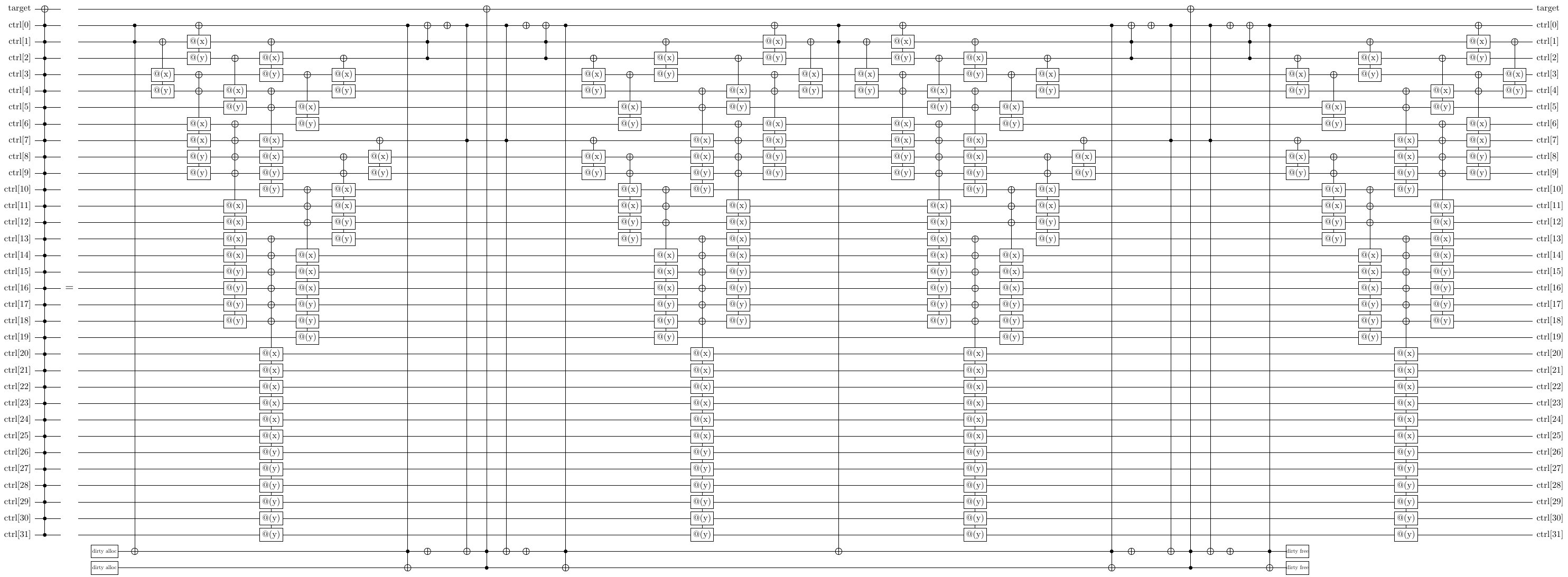}
    \caption{Example circuit for decomposing a 19-bit Toffoli into 68 Toffoli using 1 dirty ancilla and $\mathcal{O}(\log{n})$ depth.}
\label{fig:two_dirty_logn_depth:sfig2}
\end{subfigure}
\caption{Decomposition of $n$-qubit Toffoli into \mcxoptcomplexitydirty Toffoli using and $O(\log{n})$ depth using 2 dirty ancilla.}
\label{fig:two_dirty_logn_depth}
\end{adjustwidth}
\end{figure}

\section{Efficient Incrementer and Comparator Circuits} \label{sec:produce_consume_all_and}

In this section, we leverage the concept of conditionally clean ancillae to construct efficient circuits for incrementing a quantum register and for comparing a quantum register to a classical constant.  Both of these constructions rely on efficiently producing and consuming a ladder of prefix (or suffix) ANDs.

\subsection{Producing/Consuming all \texorpdfstring{$n$}{n}-bit prefix/suffix ANDs using \texorpdfstring{$3n$}{3n} Toffoli and \texorpdfstring{$\log_2^*{n}$}{log*(n)} clean ancilla}

A key building block for both the incrementer and comparator circuits is the ability to efficiently compute and consume all prefix (or suffix) ANDs of a set of qubits.  A \emph{prefix AND} of the first $i$ qubits $q_0, q_1, ..., q_{i-1}$ is the logical AND of those qubits: $q_0 \land q_1 \land ... \land q_{i-1}$.  A \emph{suffix AND} is similar, but operates on the last $i$ qubits.

If we had access to $n$ clean ancilla qubits, one could build a ladder of $n$ $\text{AND} / \text{AND}^{\dagger}$ gates such that the $i$'th ancilla qubit stores the prefix / suffix $\text{AND}$ of first / last $i$ qubits. Then, we can consume each prefix / suffix $\text{AND}$ using a $\text{CNOT}$ gate controlled on the $i$'th ancilla. This gives us a decomposition which uses $2n$ $\text{Toffoli}$ (or $n$ pairs of $\text{AND} / \text{AND}^{\dagger}$) gates to compute / uncompute every prefix / suffix AND and we can consume them using a single $\text{CNOT}$ gate.

With the help of conditionally clean ancilla qubits, we will aim to achieve the same thing - i.e. compute and uncompute the prefix / suffix $\text{AND}$ of all of the $n$ qubits on $n$ different conditionally clean qubits using $2n$ Toffoli gates and consume each prefix/suffix $\text{AND}$ using a $\text{Toffoli}$ gate. This gives us an overall Toffoli complexity of $3n$ instead of $2n$, like the clean ancilla case described above, because consuming a prefix $\text{AND}$ stored on a conditionally clean ancilla requires a $\text{Toffoli}$ gate instead of a $\text{CNOT}$ gate. 

The construction is shown in \fig{incrementer_3n} and follows closely from the construction described in \sec{two_clean_logn_depth} but instead of processing the $i$'th batch of size $2^{i}$ in log-depth, we now process it in linear depth and compute the prefix AND using $2^{i} - 1$ conditionally clean ancilla qubits as temporary workspace. 

Also, this time we need to recursively apply the same decomposition that computes the prefix $\text{AND}$ over the $K=\log{n}$ unmarked items obtained after the first decomposition because we wish to consume the prefix $\text{AND}$s sequentially and thus we want to build the structure recursively.
\fig{incrementer_3n:small} shows the circuit structure after only the first level of decomposition where consuming the prefix ANDs requires one to apply a $\log{n}$-controlled $\text{MCX}$ instead of a $\text{Toffoli}$. \fig{incrementer_3n:large} shows the full recursive decomposition where the $\log{n}$-controlled $\text{MCX}$ are recursively decomposed until $\log_2^*{n}\leq 1$ using the same strategy so that we can access the prefix AND of every prefix by controlling on at-most 2 qubits using a Toffoli gate. 

For all practical purposes, the required number of ancilla qubits is 5. Requiring six ancilla qubits would mean you were attempting to increment an integer with more digits than atoms in the observable universe.

\subsection{\texorpdfstring{$n$}{n}-bit Incrementer into \texorpdfstring{$3n$}{3n} Toffoli and \texorpdfstring{$O(n)$}{O(n)} depth using \texorpdfstring{$\log_2^*{n}$}{log*(n)} clean ancilla}\label{sec:n_bit_incrementer}

An $n$-bit incrementer adds 1 (modulo $2^n$) to an $n$-bit quantum register.  A naive implementation, shown in \fig{incrementer_3n:naive}, uses $n$ different MCX gates, where the $i$-th MCX gate is controlled on the first $i-1$ qubits and targets the $i$-th qubit~\cite{gidney2015blogincrement}. The key thing to observe here is that we need to compute and consume the prefix AND of all of the first $i$ qubits one by one.  We now use the technique described above to decompose this ladder of prefix ANDs into $3n$ Toffolis using $\log_2^*{n}$ clean ancillae. The complete decomposition is shown in \fig{incrementer_3n}.

\begin{figure}
\begin{adjustwidth}{-2cm}{-2cm}
\begin{subfigure}{0.4\linewidth}
\includegraphics[width=\linewidth]{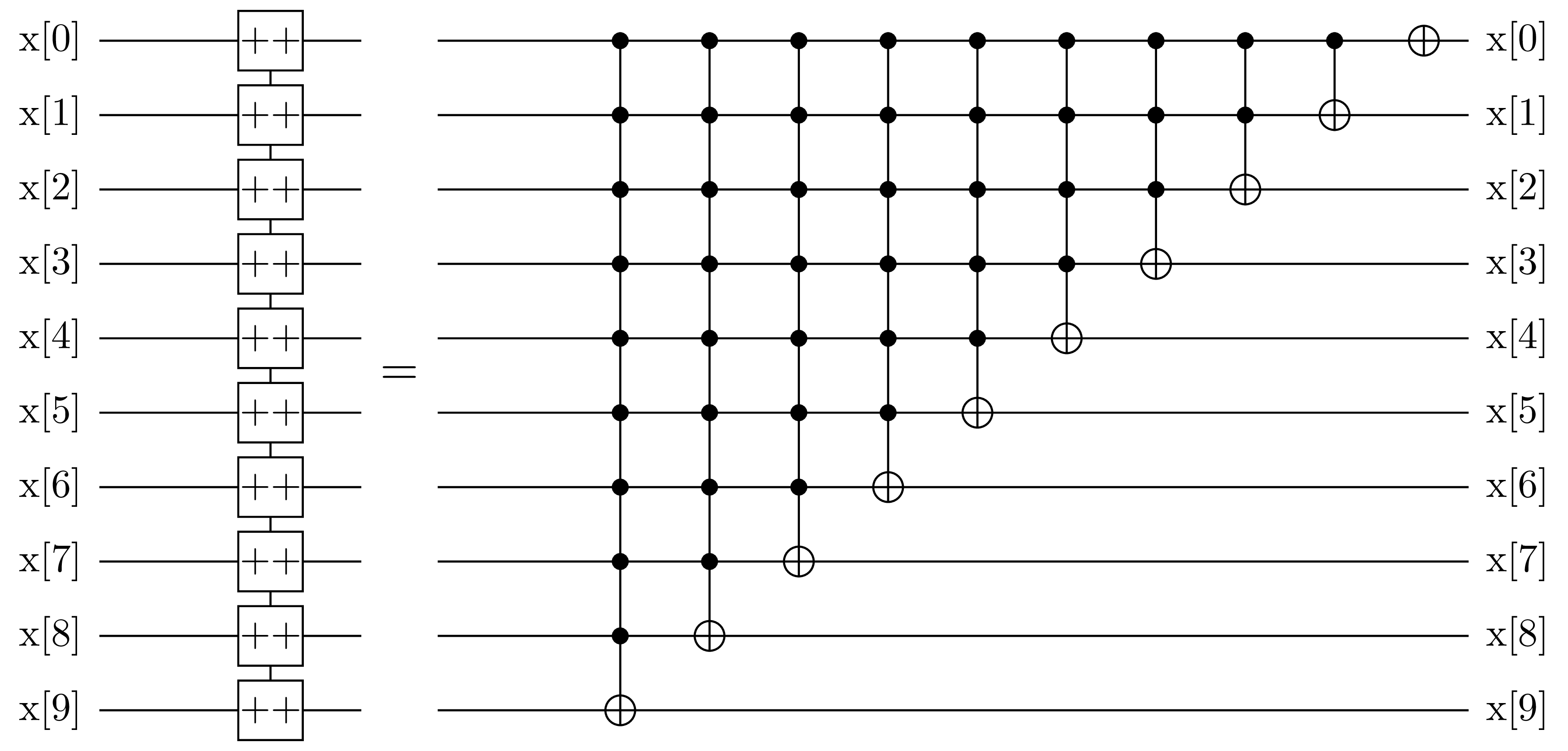}
\caption{Naive incrementer using $n$ MCX gates.}
\label{fig:incrementer_3n:naive}
\end{subfigure}~
\begin{subfigure}{0.58\linewidth}
\includegraphics[width=\linewidth]{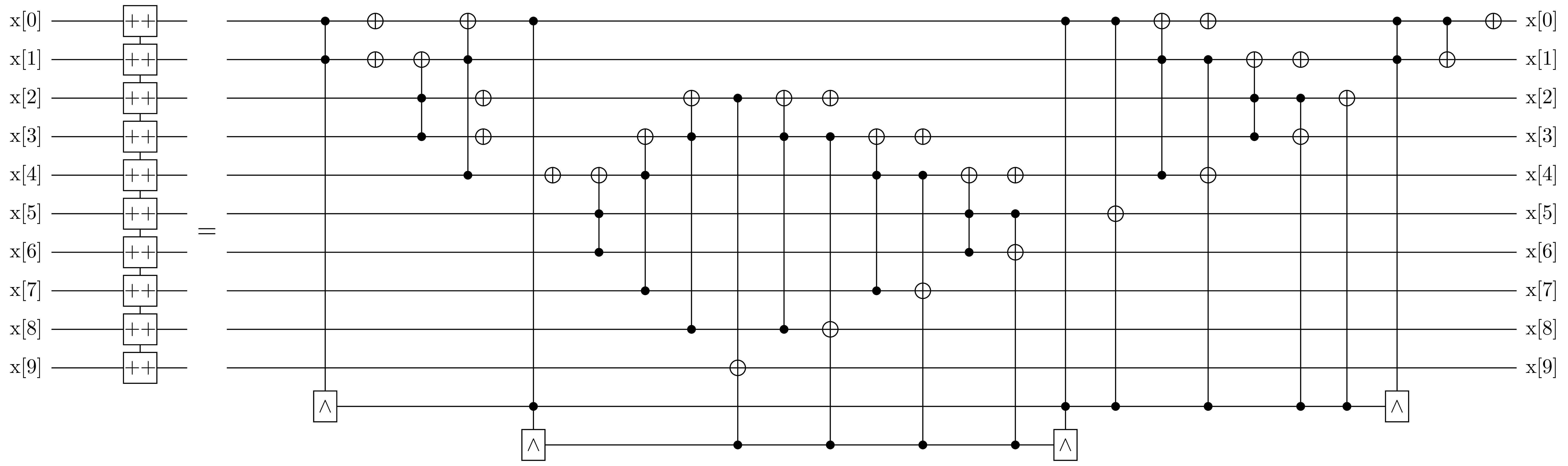}
\caption{Incrementer using conditionally clean qubits and 3n Toffoli.}
\label{fig:incrementer_3n:small}
\end{subfigure}
\begin{subfigure}[]{\linewidth}
\includegraphics[width=0.95\linewidth, height=180pt, left]{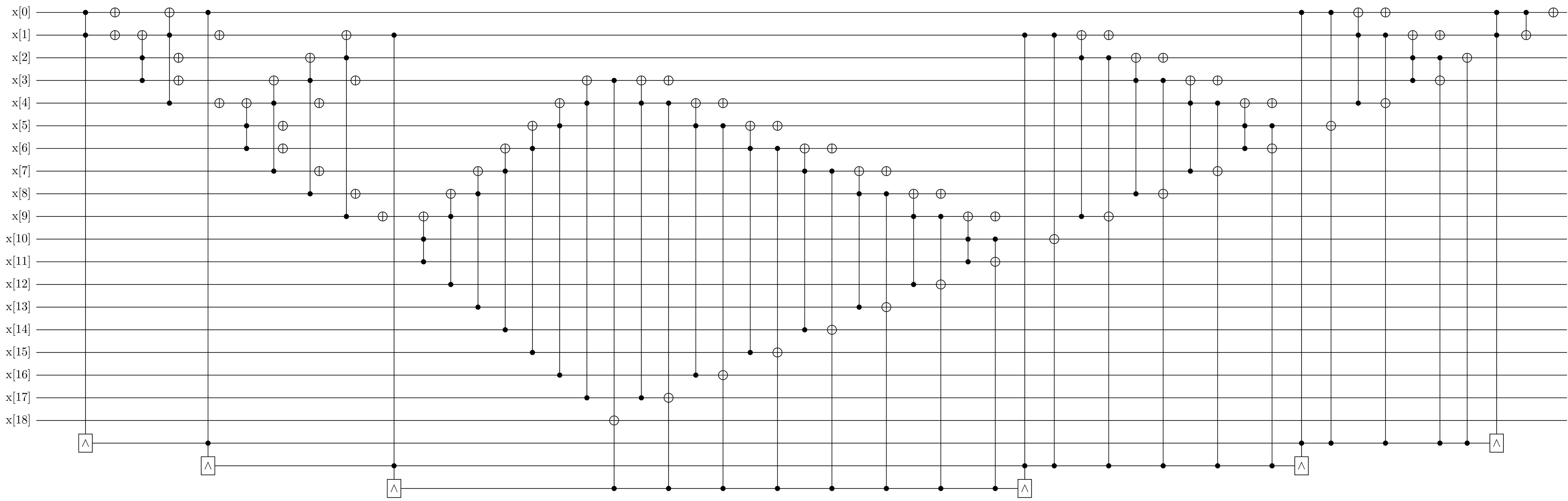}
\caption{$19$ qubit Incrementer recursively decomposed to produce/consume prefix AND for all of the first $18$ prefixes using $3n$ Toffoli and 3 clean ancillae.}
\label{fig:incrementer_3n:large}
\end{subfigure}
\caption{Decomposition of $n$-bit Incrementer into $3n$ Toffoli and $O(n)$ depth using $\log_2^*{n}$ clean ancilla.}
\label{fig:incrementer_3n}
\end{adjustwidth}
\end{figure}

\subsection{\texorpdfstring{$n$}{n}-bit LessThanConst into \texorpdfstring{$3n$}{3n} Toffoli and \texorpdfstring{$O(n)$}{O(n)} depth using \texorpdfstring{$\log_2^*{n}$}{log*(n)} clean ancilla}\label{sec:n_bit_comparator}
A $\text{LessThanConst}$ primitive implements a Quantum-Classical comparison of the form 
$$
\text{LessThanConst}_{c} \ket{x} \ket{t} \rightarrow \ket{x} \ket{t \oplus (x < c)}
$$
where $x$ is an $n$-bit quantum register, $c$ is a classical $n$-bit constant, and $t$ is a target qubit.

We can implement this operation by reducing it to the problem of consuming a ladder of suffix ANDs. The number of suffix ANDs to be consumed depends on the bits of the classical constant $c$.  In the worst case, we consume all $n$ suffix ANDs. \fig{less_than_const:quirk} shows how this reduction works.

Once we have reduced the problem to consuming suffix ANDs, we can directly apply the technique described earlier in this section to decompose the ladder of suffix ANDs into $3n$ Toffolis using $\log_2^*{n}$ clean qubits. The complete decomposition for a specific constant $c$ is shown in \fig{less_than_const:large}.

\begin{figure}
\begin{adjustwidth}{-2cm}{-2cm}
\begin{subfigure}{\linewidth}
\includegraphics[width=0.95\linewidth, left]{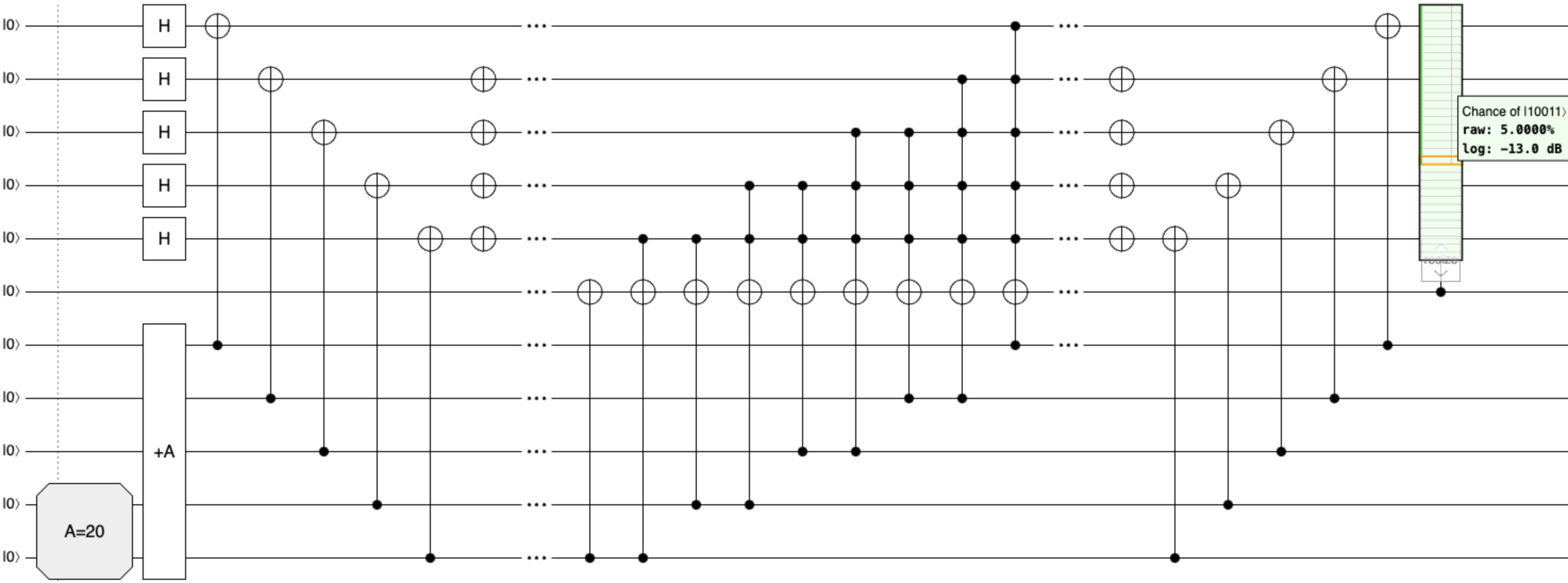}
\caption{Naive LessThanConst using $n$ MCX gates. Here is a \href{https://algassert.com/quirk\#circuit=\%7B\%22cols\%22\%3A\%5B\%5B\%22H\%22\%2C\%22H\%22\%2C\%22H\%22\%2C\%22H\%22\%2C\%22H\%22\%2C\%22H\%22\%2C1\%2C\%22Counting6\%22\%5D\%2C\%5B\%22X\%22\%2C1\%2C1\%2C1\%2C1\%2C1\%2C1\%2C\%22\%E2\%80\%A2\%22\%5D\%2C\%5B1\%2C\%22X\%22\%2C1\%2C1\%2C1\%2C1\%2C1\%2C1\%2C\%22\%E2\%80\%A2\%22\%5D\%2C\%5B1\%2C1\%2C\%22X\%22\%2C1\%2C1\%2C1\%2C1\%2C1\%2C1\%2C\%22\%E2\%80\%A2\%22\%5D\%2C\%5B1\%2C1\%2C1\%2C\%22X\%22\%2C1\%2C1\%2C1\%2C1\%2C1\%2C1\%2C\%22\%E2\%80\%A2\%22\%5D\%2C\%5B1\%2C1\%2C1\%2C1\%2C\%22X\%22\%2C1\%2C1\%2C1\%2C1\%2C1\%2C1\%2C\%22\%E2\%80\%A2\%22\%5D\%2C\%5B1\%2C1\%2C1\%2C1\%2C1\%2C\%22X\%22\%2C1\%2C1\%2C1\%2C1\%2C1\%2C1\%2C\%22\%E2\%80\%A2\%22\%5D\%2C\%5B1\%2C\%22X\%22\%2C\%22X\%22\%2C\%22X\%22\%2C\%22X\%22\%2C\%22X\%22\%5D\%2C\%5B\%22\%E2\%80\%A6\%22\%2C\%22\%E2\%80\%A6\%22\%2C\%22\%E2\%80\%A6\%22\%2C\%22\%E2\%80\%A6\%22\%2C\%22\%E2\%80\%A6\%22\%2C\%22\%E2\%80\%A6\%22\%2C\%22\%E2\%80\%A6\%22\%5D\%2C\%5B1\%2C1\%2C1\%2C1\%2C1\%2C1\%2C\%22X\%22\%2C1\%2C1\%2C1\%2C1\%2C1\%2C\%22\%E2\%80\%A2\%22\%5D\%2C\%5B1\%2C1\%2C1\%2C1\%2C1\%2C\%22\%E2\%80\%A2\%22\%2C\%22X\%22\%2C1\%2C1\%2C1\%2C1\%2C1\%2C\%22\%E2\%80\%A2\%22\%5D\%2C\%5B1\%2C1\%2C1\%2C1\%2C1\%2C\%22\%E2\%80\%A2\%22\%2C\%22X\%22\%2C1\%2C1\%2C1\%2C1\%2C\%22\%E2\%80\%A2\%22\%5D\%2C\%5B1\%2C1\%2C1\%2C1\%2C\%22\%E2\%80\%A2\%22\%2C\%22\%E2\%80\%A2\%22\%2C\%22X\%22\%2C1\%2C1\%2C1\%2C1\%2C\%22\%E2\%80\%A2\%22\%5D\%2C\%5B1\%2C1\%2C1\%2C1\%2C\%22\%E2\%80\%A2\%22\%2C\%22\%E2\%80\%A2\%22\%2C\%22X\%22\%2C1\%2C1\%2C1\%2C\%22\%E2\%80\%A2\%22\%5D\%2C\%5B1\%2C1\%2C1\%2C\%22\%E2\%80\%A2\%22\%2C\%22\%E2\%80\%A2\%22\%2C\%22\%E2\%80\%A2\%22\%2C\%22X\%22\%2C1\%2C1\%2C1\%2C\%22\%E2\%80\%A2\%22\%5D\%2C\%5B1\%2C1\%2C1\%2C\%22\%E2\%80\%A2\%22\%2C\%22\%E2\%80\%A2\%22\%2C\%22\%E2\%80\%A2\%22\%2C\%22X\%22\%2C1\%2C1\%2C\%22\%E2\%80\%A2\%22\%5D\%2C\%5B1\%2C1\%2C\%22\%E2\%80\%A2\%22\%2C\%22\%E2\%80\%A2\%22\%2C\%22\%E2\%80\%A2\%22\%2C\%22\%E2\%80\%A2\%22\%2C\%22X\%22\%2C1\%2C1\%2C\%22\%E2\%80\%A2\%22\%5D\%2C\%5B1\%2C1\%2C\%22\%E2\%80\%A2\%22\%2C\%22\%E2\%80\%A2\%22\%2C\%22\%E2\%80\%A2\%22\%2C\%22\%E2\%80\%A2\%22\%2C\%22X\%22\%2C1\%2C\%22\%E2\%80\%A2\%22\%5D\%2C\%5B1\%2C\%22\%E2\%80\%A2\%22\%2C\%22\%E2\%80\%A2\%22\%2C\%22\%E2\%80\%A2\%22\%2C\%22\%E2\%80\%A2\%22\%2C\%22\%E2\%80\%A2\%22\%2C\%22X\%22\%2C1\%2C\%22\%E2\%80\%A2\%22\%5D\%2C\%5B\%22\%E2\%80\%A2\%22\%2C\%22\%E2\%80\%A2\%22\%2C\%22\%E2\%80\%A2\%22\%2C\%22\%E2\%80\%A2\%22\%2C\%22\%E2\%80\%A2\%22\%2C\%22\%E2\%80\%A2\%22\%2C\%22X\%22\%2C\%22\%E2\%80\%A2\%22\%5D\%2C\%5B\%22\%E2\%80\%A6\%22\%2C\%22\%E2\%80\%A6\%22\%2C\%22\%E2\%80\%A6\%22\%2C\%22\%E2\%80\%A6\%22\%2C\%22\%E2\%80\%A6\%22\%2C\%22\%E2\%80\%A6\%22\%2C\%22\%E2\%80\%A6\%22\%5D\%2C\%5B1\%2C\%22X\%22\%2C\%22X\%22\%2C\%22X\%22\%2C\%22X\%22\%2C\%22X\%22\%5D\%2C\%5B1\%2C1\%2C1\%2C1\%2C1\%2C\%22X\%22\%2C1\%2C1\%2C1\%2C1\%2C1\%2C1\%2C\%22\%E2\%80\%A2\%22\%5D\%2C\%5B1\%2C1\%2C1\%2C1\%2C\%22X\%22\%2C1\%2C1\%2C1\%2C1\%2C1\%2C1\%2C\%22\%E2\%80\%A2\%22\%5D\%2C\%5B1\%2C1\%2C1\%2C\%22X\%22\%2C1\%2C1\%2C1\%2C1\%2C1\%2C1\%2C\%22\%E2\%80\%A2\%22\%5D\%2C\%5B1\%2C1\%2C\%22X\%22\%2C1\%2C1\%2C1\%2C1\%2C1\%2C1\%2C\%22\%E2\%80\%A2\%22\%5D\%2C\%5B1\%2C\%22X\%22\%2C1\%2C1\%2C1\%2C1\%2C1\%2C1\%2C\%22\%E2\%80\%A2\%22\%5D\%2C\%5B\%22X\%22\%2C1\%2C1\%2C1\%2C1\%2C1\%2C1\%2C\%22\%E2\%80\%A2\%22\%5D\%2C\%5B\%22Chance6\%22\%2C1\%2C1\%2C1\%2C1\%2C1\%2C\%22\%E2\%80\%A2\%22\%5D\%5D\%7D}{quirk link} of a circuit which shows this reduction.}
\label{fig:less_than_const:quirk}
\end{subfigure}
\begin{subfigure}[]{\linewidth}
\includegraphics[width=0.95\linewidth, height=170pt]{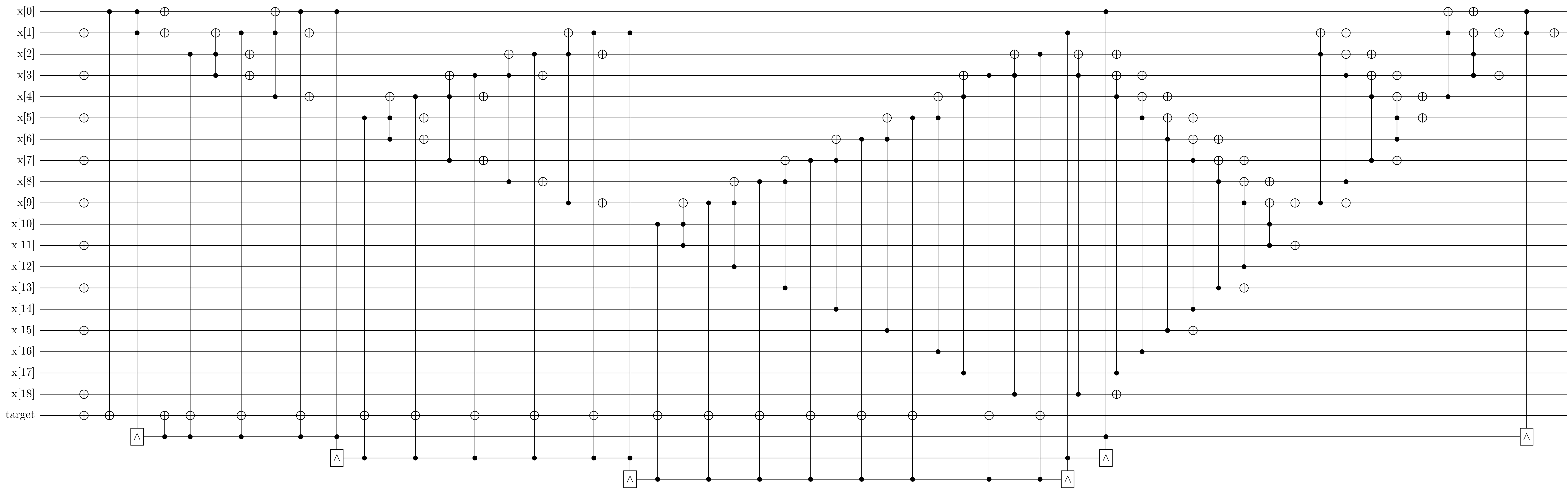}
\caption{$19$-bit $\text{LessThanConst}$ for constant $c=349525$ that has a binary representation $1010101010101010101$. The quantum-classical comparator uses 3 clean ancillae and $47$ Toffoli gates.}
\label{fig:less_than_const:large}
\end{subfigure}
\caption{Decomposition of $n$-bit LessThanConst into $3n$ Toffoli and $O(n)$ depth using $\log_2^*{n}$ clean ancilla.}
\label{fig:less_than_const}
\end{adjustwidth}
\end{figure}

\section{Constructions for Unary Iteration and QROM}\label{sec:unary_iteration}

Given an $n$-bit selection register $S$ and a sequence of unitary operations $[V_0, V_1, \dots, V_{2^n - 1}]$, each acting on an $m$-bit target register; Unary iteration is a technique that allows us to apply $V_i$ to the target register when the selection register stores integer $\ket{i}$. In other words, the action of unary iteration can be described by
\begin{equation}\label{eq:unary_iteration}
    U = \sum_{i=0}^{N-1}\ket{i}\bra{i} \otimes V_{i}
\end{equation}
where $N=2^n$ is the number of unitaries $V_i$ to apply on different branches of the superposition. 

The technique was introduced in~\cite{Babbush2018} and now forms a  fundamental building block for many quantum algorithms~\cite{Lee2021, Rubin2024, yuan2023improved}, including the Quantum Read-Only Memory (QROM) operation~\cite{zhu2024unifiedarchitecturequantumlookup}. QROM is a technique for loading classical data into a quantum superposition.

The Controlled Unary Iteration construction introduced by~\cite{Babbush2018} uses $n$ clean ancilla qubits to iterate on $N = 2^{n}$ indices with a Toffoli cost of $N - 1$. 
In this section, we will describe two new constructions for the low ancilla regime where we either consume a constant number of clean ancillae and use the system qubits as conditionally clean qubits, or we borrow $n$ dirty qubits from the system.

\subsection{Unary Iteration as a tree traversal}
Unary iteration can be viewed as a tree traversal where each node of the tree corresponds to one or more elements of the range we wish to iterate upon and traversing the tree in a DFS order produces a circuit that corresponds to unary iteration. 

In \fig{qrom_binary_tree} we explain how one can view the standard unary iteration circuit as a tree traversal of a balanced binary tree. If we are given $n$ clean ancillae, we can use the circuit elements from \fig{qrom_binary_tree:sfig3} to arrive at the construction with a Toffoli cost of $N - 1$, presented in~\cite{Babbush2018}.

In \fig{qrom_binary_tree:sfig4} and \fig{qrom_binary_tree:sfig5} we show how we can modify the construction to avoid measurement based uncomputation with a Toffoli cost of $1.5N - 1$. This construction is identical to the one described in Appendix G.4 of~\cite{Childs2018}. Avoiding measurement based uncomputation will become important when we describe constructions using conditionally clean ancillae or dirty ancillae - both of which cannot be cleaned up using measurement based uncomputation. 

\subsection{Skew Trees for Unary Iteration}

In situations where it's cheap to apply the inverse of operations we are indexing over, we can use a different tree than balanced binary trees. In \fig{qrom_skew_tree:sfig1}, we give a recursive definition of a ``Skew tree" and in \fig{qrom_skew_tree:sfig2} we show how tree traversal of a Skew tree can be translated into unary iteration. 
The basic idea of the skewed tree is to replace circuits of the form ``if $\lnot C$ then do $A$; if $C$ then do $B$" with circuits of the form ``do $A$; if $C$ then do $A^{-1} \cdot B$".
Instead of doing $A$ conditionally, we do $A$ unconditionally but undo $A$ in addition to doing $B$ when $C$ is true.
This optimization works best when $A^{-1} \cdot B$ is efficient to apply.
This is the case in QROM reads, where $A$ and $B$ are both a product of Pauli X gates. The skewed tree has better behavior when $L$ is not a power of 2.
For example, the lopsided tree over $L$ items has $\lfloor \log_2{L} \rfloor$ levels, whereas the binary tree has $\lceil \log_2{L} \rceil$ levels.
This saves an ancilla qubit for domain sizes $L$ that are not a power of 2.
The skewed tree associates outputs with every node, instead of only associating outputs with leaf nodes.
This halves the size of the tree, which reduces the size of the circuit.
In \fig{qrom_skew_tree:sfig4}, we show how one can use a Skew tree to perform unary iteration of $N$ elements using $1.25N - 1$ Toffoli gates and still avoid measurement based uncomputation (compared to $1.5N - 1$ using balanced binary trees).  

Here is a \href{https://algassert.com/quirk#circuit=%7B%22cols%22%3A%5B%5B1%2C1%2C1%2C1%2C%22Counting4%22%5D%2C%5B1%2C1%2C1%2C%22%E2%80%A2%22%2C1%2C1%2C1%2C1%2C%22X%22%5D%2C%5B1%2C1%2C1%2C%22%E2%80%A2%22%2C1%2C%22Chance4%22%5D%2C%5B1%2C1%2C1%2C%22%E2%80%A2%22%2C%22%E2%80%A2%22%2C1%2C1%2C1%2C%22X%22%2C%22X%22%5D%2C%5B1%2C1%2C%22X%22%2C%22%E2%80%A2%22%2C1%2C%22%E2%80%A2%22%5D%2C%5B1%2C1%2C%22%E2%80%A2%22%2C1%2C%22Chance4%22%5D%2C%5B1%2C1%2C%22%E2%80%A2%22%2C1%2C1%2C1%2C1%2C1%2C1%2C%22X%22%5D%2C%5B1%2C1%2C%22%E2%80%A2%22%2C1%2C%22%E2%80%A2%22%2C1%2C1%2C1%2C1%2C1%2C%22X%22%5D%2C%5B1%2C1%2C1%2C1%2C1%2C%22X%22%2C%22%E2%80%A2%22%5D%2C%5B1%2C1%2C%22X%22%2C%22%E2%80%A2%22%2C1%2C%22%E2%80%A2%22%5D%2C%5B1%2C1%2C1%2C1%2C1%2C%22X%22%2C%22%E2%80%A2%22%5D%2C%5B1%2C1%2C%22%E2%80%A2%22%2C1%2C%22Chance4%22%5D%2C%5B1%2C1%2C%22%E2%80%A2%22%2C1%2C1%2C1%2C1%2C1%2C1%2C1%2C%22X%22%5D%2C%5B1%2C1%2C%22%E2%80%A2%22%2C1%2C%22%E2%80%A2%22%2C1%2C1%2C1%2C%22%E2%80%A6%22%2C%22%E2%80%A6%22%2C%22%E2%80%A6%22%2C%22%E2%80%A6%22%5D%2C%5B1%2C%22X%22%2C%22%E2%80%A2%22%2C1%2C1%2C%22%E2%80%A2%22%5D%2C%5B1%2C%22%E2%80%A2%22%2C1%2C1%2C%22Chance4%22%5D%2C%5B1%2C%22%E2%80%A2%22%2C1%2C1%2C1%2C1%2C1%2C1%2C%22%E2%80%A6%22%2C%22%E2%80%A6%22%2C%22%E2%80%A6%22%2C%22%E2%80%A6%22%5D%2C%5B1%2C%22%E2%80%A2%22%2C1%2C1%2C%22%E2%80%A2%22%2C1%2C1%2C1%2C1%2C1%2C1%2C%22X%22%5D%2C%5B1%2C%22X%22%2C%22%E2%80%A2%22%2C1%2C1%2C%22%E2%80%A2%22%5D%2C%5B1%2C1%2C1%2C1%2C1%2C1%2C%22X%22%2C%22%E2%80%A2%22%5D%2C%5B1%2C1%2C%22X%22%2C%22%E2%80%A2%22%2C1%2C1%2C%22%E2%80%A2%22%5D%2C%5B1%2C1%2C1%2C1%2C1%2C1%2C%22X%22%2C%22%E2%80%A2%22%5D%2C%5B1%2C1%2C%22%E2%80%A2%22%2C1%2C%22Chance4%22%5D%2C%5B1%2C1%2C%22%E2%80%A2%22%2C1%2C1%2C1%2C1%2C1%2C1%2C1%2C1%2C%22X%22%5D%2C%5B1%2C1%2C%22%E2%80%A2%22%2C1%2C%22%E2%80%A2%22%2C1%2C1%2C1%2C%22%E2%80%A6%22%2C%22%E2%80%A6%22%2C%22%E2%80%A6%22%2C%22%E2%80%A6%22%5D%2C%5B1%2C%22X%22%2C%22%E2%80%A2%22%2C1%2C1%2C%22%E2%80%A2%22%5D%2C%5B1%2C%22%E2%80%A2%22%2C1%2C1%2C%22Chance4%22%5D%2C%5B1%2C%22%E2%80%A2%22%2C1%2C1%2C1%2C1%2C1%2C1%2C%22%E2%80%A6%22%2C%22%E2%80%A6%22%2C%22%E2%80%A6%22%2C%22%E2%80%A6%22%5D%2C%5B1%2C%22%E2%80%A2%22%2C1%2C1%2C%22%E2%80%A2%22%2C1%2C1%2C1%2C%22%E2%80%A6%22%2C%22%E2%80%A6%22%2C%22%E2%80%A6%22%2C%22%E2%80%A6%22%5D%2C%5B1%2C1%2C1%2C1%2C1%2C%22X%22%2C%22%E2%80%A2%22%5D%2C%5B1%2C%22X%22%2C%22%E2%80%A2%22%2C1%2C1%2C%22%E2%80%A2%22%5D%2C%5B1%2C1%2C1%2C1%2C1%2C%22X%22%2C%22%E2%80%A2%22%5D%2C%5B1%2C%22%E2%80%A2%22%2C1%2C1%2C%22Chance4%22%5D%2C%5B1%2C%22%E2%80%A2%22%2C1%2C1%2C1%2C1%2C1%2C1%2C1%2C%22%E2%80%A6%22%2C%22%E2%80%A6%22%2C%22%E2%80%A6%22%5D%2C%5B1%2C%22%E2%80%A2%22%2C1%2C1%2C%22%E2%80%A2%22%2C1%2C1%2C1%2C1%2C%22%E2%80%A6%22%2C%22%E2%80%A6%22%2C%22%E2%80%A6%22%5D%2C%5B1%2C%22%E2%80%A2%22%2C1%2C%22%E2%80%A2%22%2C%22Chance4%22%5D%2C%5B%22X%22%2C%22%E2%80%A2%22%2C1%2C1%2C1%2C%22%E2%80%A2%22%5D%2C%5B%22%E2%80%A2%22%2C1%2C1%2C1%2C%22Chance4%22%5D%2C%5B%22%E2%80%A2%22%2C1%2C1%2C1%2C1%2C1%2C1%2C1%2C%22%E2%80%A6%22%2C%22%E2%80%A6%22%2C%22%E2%80%A6%22%2C%22%E2%80%A6%22%5D%2C%5B%22%E2%80%A2%22%2C1%2C1%2C1%2C%22%E2%80%A2%22%2C1%2C1%2C1%2C1%2C1%2C1%2C1%2C%22X%22%5D%2C%5B%22X%22%2C%22%E2%80%A2%22%2C1%2C1%2C1%2C%22%E2%80%A2%22%5D%2C%5B1%2C%22X%22%2C%22%E2%80%A2%22%2C1%2C1%2C1%2C%22%E2%80%A2%22%5D%2C%5B1%2C1%2C%22X%22%2C%22%E2%80%A2%22%2C1%2C1%2C1%2C%22%E2%80%A2%22%5D%2C%5B1%2C1%2C1%2C1%2C%22Chance4%22%2C1%2C1%2C1%2C%22Chance4%22%5D%5D%2C%22init%22%3A%5B0%2C0%2C0%2C1%5D%7D}{quirk link} to a QROM circuit that loads the first 16 natural numbers as $\text{data = [1, 2, ..., 16]}$ using a skew tree construction. Many of the controlled reads are 0 because for a skew tree construction, the data to be loaded needs to be modified as follows to account for undoing the unconditional read done at a previous index which affects the current index.

\begin{lstlisting}[language=Python]
skew_data = [0] * N
for i in range(N):
    for j in range(i, N):
        if i & j == i:
            skew_data[j] ^= data[i]
\end{lstlisting}

\subsection{Unary iteration and QROM using conditionally clean ancillae}\label{sec:unary_iteration_cca}
We can use the tricks developed in \sec{produce_consume_all_and} to consume a constant number of clean qubits and use the generated conditionally clean qubits to produce / consume a prefix AND ladder and combine it with unary iteration constructions given in \fig{qrom_binary_tree:sfig4} and \fig{qrom_skew_tree:sfig4} to get a constant ancilla version of unary iteration using both balanced binary trees and skewed trees. 
The main overhead with this approach is that consuming a prefix AND stored on a conditionally clean ancilla requires a $\text{Toffoli}$ instead of a $\text{CNOT}$. Therefore, we get an $N$ Toffoli overhead in both the approaches described above. 

Therefore, a constant number of clean ancillae and conditionally clean qubits, we can do unary iteration over $N$ elements using $2.5N$ Toffoli gates via balanced binary trees and $2.25N$ Toffoli gates via skewed trees. 

Here is a \href{https://algassert.com/quirk\#circuit=\%7B\%22cols\%22\%3A\%5B\%5B\%22Counting6\%22\%5D\%2C\%5B\%22QFT6\%22\%5D\%2C\%5B1\%2C1\%2C1\%2C1\%2C\%22\%E2\%97\%A6\%22\%2C\%22\%E2\%97\%A6\%22\%2C\%22X\%22\%5D\%2C\%5B1\%2C1\%2C\%22\%E2\%97\%A6\%22\%2C\%22\%E2\%97\%A6\%22\%2C\%22X\%22\%5D\%2C\%5B1\%2C\%22\%E2\%97\%A6\%22\%2C1\%2C1\%2C\%22\%E2\%80\%A2\%22\%2C\%22X\%22\%5D\%2C\%5B1\%2C1\%2C1\%2C1\%2C1\%2C\%22\%E2\%80\%A2\%22\%2C\%22\%E2\%80\%A2\%22\%2C\%22X\%22\%5D\%2C\%5B\%22\%E2\%97\%A6\%22\%2C1\%2C1\%2C1\%2C1\%2C1\%2C1\%2C\%22\%E2\%80\%A2\%22\%2C\%22X\%22\%5D\%2C\%5B\%22Chance4\%22\%2C1\%2C1\%2C1\%2C\%22Chance\%22\%2C\%22Chance\%22\%2C1\%2C1\%2C\%22\%E2\%80\%A2\%22\%5D\%2C\%5B1\%2C1\%2C1\%2C1\%2C1\%2C1\%2C1\%2C\%22\%E2\%80\%A2\%22\%2C\%22X\%22\%5D\%2C\%5B\%22Chance4\%22\%2C1\%2C1\%2C1\%2C\%22Chance\%22\%2C\%22Chance\%22\%2C1\%2C1\%2C\%22\%E2\%80\%A2\%22\%5D\%2C\%5B\%22\%E2\%80\%A2\%22\%2C1\%2C1\%2C1\%2C1\%2C1\%2C1\%2C\%22\%E2\%80\%A2\%22\%2C\%22X\%22\%5D\%2C\%5B1\%2C1\%2C1\%2C1\%2C\%22\%E2\%80\%A2\%22\%2C\%22X\%22\%5D\%2C\%5B1\%2C1\%2C1\%2C1\%2C\%22\%E2\%80\%A2\%22\%2C1\%2C\%22\%E2\%80\%A2\%22\%2C\%22X\%22\%5D\%2C\%5B\%22\%E2\%97\%A6\%22\%2C1\%2C1\%2C1\%2C1\%2C1\%2C1\%2C\%22\%E2\%80\%A2\%22\%2C\%22X\%22\%5D\%2C\%5B\%22Chance4\%22\%2C1\%2C1\%2C1\%2C\%22Chance\%22\%2C\%22Chance\%22\%2C1\%2C1\%2C\%22\%E2\%80\%A2\%22\%5D\%2C\%5B1\%2C1\%2C1\%2C1\%2C1\%2C1\%2C1\%2C\%22\%E2\%80\%A2\%22\%2C\%22X\%22\%5D\%2C\%5B\%22Chance4\%22\%2C1\%2C1\%2C1\%2C\%22Chance\%22\%2C\%22Chance\%22\%2C1\%2C1\%2C\%22\%E2\%80\%A2\%22\%5D\%2C\%5B\%22\%E2\%80\%A2\%22\%2C1\%2C1\%2C1\%2C1\%2C1\%2C1\%2C\%22\%E2\%80\%A2\%22\%2C\%22X\%22\%5D\%2C\%5B1\%2C1\%2C1\%2C1\%2C1\%2C\%22\%E2\%80\%A2\%22\%2C\%22\%E2\%80\%A2\%22\%2C\%22X\%22\%5D\%2C\%5B1\%2C1\%2C1\%2C\%22\%E2\%97\%A6\%22\%2C\%22X\%22\%5D\%2C\%5B1\%2C1\%2C1\%2C1\%2C\%22\%E2\%80\%A2\%22\%2C\%22X\%22\%5D\%2C\%5B1\%2C\%22\%E2\%80\%A2\%22\%2C1\%2C\%22\%E2\%97\%A6\%22\%2C1\%2C\%22X\%22\%5D\%2C\%5B1\%2C1\%2C1\%2C1\%2C1\%2C\%22\%E2\%80\%A2\%22\%2C\%22\%E2\%80\%A2\%22\%2C\%22X\%22\%5D\%2C\%5B\%22\%E2\%97\%A6\%22\%2C1\%2C1\%2C1\%2C1\%2C1\%2C1\%2C\%22\%E2\%80\%A2\%22\%2C\%22X\%22\%5D\%2C\%5B\%22Chance4\%22\%2C1\%2C1\%2C1\%2C\%22Chance\%22\%2C\%22Chance\%22\%2C1\%2C1\%2C\%22\%E2\%80\%A2\%22\%5D\%2C\%5B1\%2C1\%2C1\%2C1\%2C1\%2C1\%2C1\%2C\%22\%E2\%80\%A2\%22\%2C\%22X\%22\%5D\%2C\%5B\%22Chance4\%22\%2C1\%2C1\%2C1\%2C\%22Chance\%22\%2C\%22Chance\%22\%2C1\%2C1\%2C\%22\%E2\%80\%A2\%22\%5D\%2C\%5B\%22\%E2\%80\%A2\%22\%2C1\%2C1\%2C1\%2C1\%2C1\%2C1\%2C\%22\%E2\%80\%A2\%22\%2C\%22X\%22\%5D\%2C\%5B1\%2C1\%2C1\%2C1\%2C\%22\%E2\%80\%A2\%22\%2C\%22X\%22\%5D\%2C\%5B1\%2C1\%2C1\%2C1\%2C\%22\%E2\%80\%A2\%22\%2C1\%2C\%22\%E2\%80\%A2\%22\%2C\%22X\%22\%5D\%2C\%5B\%22\%E2\%97\%A6\%22\%2C1\%2C1\%2C1\%2C1\%2C1\%2C1\%2C\%22\%E2\%80\%A2\%22\%2C\%22X\%22\%5D\%2C\%5B\%22Chance4\%22\%2C1\%2C1\%2C1\%2C\%22Chance\%22\%2C\%22Chance\%22\%2C1\%2C1\%2C\%22\%E2\%80\%A2\%22\%5D\%2C\%5B1\%2C1\%2C1\%2C1\%2C1\%2C1\%2C1\%2C\%22\%E2\%80\%A2\%22\%2C\%22X\%22\%5D\%2C\%5B\%22Chance4\%22\%2C1\%2C1\%2C1\%2C\%22Chance\%22\%2C\%22Chance\%22\%2C1\%2C1\%2C\%22\%E2\%80\%A2\%22\%5D\%2C\%5B\%22\%E2\%80\%A2\%22\%2C1\%2C1\%2C1\%2C1\%2C1\%2C1\%2C\%22\%E2\%80\%A2\%22\%2C\%22X\%22\%5D\%2C\%5B1\%2C1\%2C1\%2C1\%2C1\%2C\%22\%E2\%80\%A2\%22\%2C\%22\%E2\%80\%A2\%22\%2C\%22X\%22\%5D\%2C\%5B1\%2C1\%2C\%22\%E2\%80\%A2\%22\%2C\%22X\%22\%5D\%2C\%5B1\%2C1\%2C1\%2C\%22\%E2\%80\%A2\%22\%2C\%22X\%22\%5D\%2C\%5B1\%2C1\%2C1\%2C1\%2C\%22\%E2\%80\%A2\%22\%2C\%22X\%22\%5D\%2C\%5B1\%2C\%22\%E2\%80\%A2\%22\%2C1\%2C\%22\%E2\%80\%A2\%22\%2C1\%2C\%22X\%22\%5D\%2C\%5B1\%2C1\%2C\%22\%E2\%80\%A2\%22\%2C\%22X\%22\%5D\%2C\%5B1\%2C1\%2C1\%2C1\%2C1\%2C\%22\%E2\%80\%A2\%22\%2C\%22\%E2\%80\%A2\%22\%2C\%22X\%22\%5D\%2C\%5B\%22\%E2\%97\%A6\%22\%2C1\%2C1\%2C1\%2C1\%2C1\%2C1\%2C\%22\%E2\%80\%A2\%22\%2C\%22X\%22\%5D\%2C\%5B\%22Chance4\%22\%2C1\%2C1\%2C1\%2C\%22Chance\%22\%2C\%22Chance\%22\%2C1\%2C1\%2C\%22\%E2\%80\%A2\%22\%5D\%2C\%5B1\%2C1\%2C1\%2C1\%2C1\%2C1\%2C1\%2C\%22\%E2\%80\%A2\%22\%2C\%22X\%22\%5D\%2C\%5B\%22Chance4\%22\%2C1\%2C1\%2C1\%2C\%22Chance\%22\%2C\%22Chance\%22\%2C1\%2C1\%2C\%22\%E2\%80\%A2\%22\%5D\%2C\%5B\%22\%E2\%80\%A2\%22\%2C1\%2C1\%2C1\%2C1\%2C1\%2C1\%2C\%22\%E2\%80\%A2\%22\%2C\%22X\%22\%5D\%2C\%5B1\%2C1\%2C1\%2C1\%2C\%22\%E2\%80\%A2\%22\%2C\%22X\%22\%5D\%2C\%5B1\%2C1\%2C1\%2C1\%2C\%22\%E2\%80\%A2\%22\%2C1\%2C\%22\%E2\%80\%A2\%22\%2C\%22X\%22\%5D\%2C\%5B\%22\%E2\%97\%A6\%22\%2C1\%2C1\%2C1\%2C1\%2C1\%2C1\%2C\%22\%E2\%80\%A2\%22\%2C\%22X\%22\%5D\%2C\%5B\%22Chance4\%22\%2C1\%2C1\%2C1\%2C\%22Chance\%22\%2C\%22Chance\%22\%2C1\%2C1\%2C\%22\%E2\%80\%A2\%22\%5D\%2C\%5B1\%2C1\%2C1\%2C1\%2C1\%2C1\%2C1\%2C\%22\%E2\%80\%A2\%22\%2C\%22X\%22\%5D\%2C\%5B\%22Chance4\%22\%2C1\%2C1\%2C1\%2C\%22Chance\%22\%2C\%22Chance\%22\%2C1\%2C1\%2C\%22\%E2\%80\%A2\%22\%5D\%2C\%5B\%22\%E2\%80\%A2\%22\%2C1\%2C1\%2C1\%2C1\%2C1\%2C1\%2C\%22\%E2\%80\%A2\%22\%2C\%22X\%22\%5D\%2C\%5B1\%2C1\%2C1\%2C1\%2C1\%2C\%22\%E2\%80\%A2\%22\%2C\%22\%E2\%80\%A2\%22\%2C\%22X\%22\%5D\%2C\%5B1\%2C1\%2C1\%2C\%22\%E2\%80\%A2\%22\%2C\%22X\%22\%5D\%2C\%5B1\%2C1\%2C1\%2C1\%2C\%22\%E2\%80\%A2\%22\%2C\%22X\%22\%5D\%2C\%5B1\%2C\%22\%E2\%80\%A2\%22\%2C1\%2C\%22\%E2\%80\%A2\%22\%2C1\%2C\%22X\%22\%5D\%2C\%5B1\%2C1\%2C1\%2C1\%2C1\%2C\%22\%E2\%80\%A2\%22\%2C\%22\%E2\%80\%A2\%22\%2C\%22X\%22\%5D\%2C\%5B\%22\%E2\%97\%A6\%22\%2C1\%2C1\%2C1\%2C1\%2C1\%2C1\%2C\%22\%E2\%80\%A2\%22\%2C\%22X\%22\%5D\%2C\%5B\%22Chance4\%22\%2C1\%2C1\%2C1\%2C\%22Chance\%22\%2C\%22Chance\%22\%2C1\%2C1\%2C\%22\%E2\%80\%A2\%22\%5D\%2C\%5B1\%2C1\%2C1\%2C1\%2C1\%2C1\%2C1\%2C\%22\%E2\%80\%A2\%22\%2C\%22X\%22\%5D\%2C\%5B\%22Chance4\%22\%2C1\%2C1\%2C1\%2C\%22Chance\%22\%2C\%22Chance\%22\%2C1\%2C1\%2C\%22\%E2\%80\%A2\%22\%5D\%2C\%5B\%22\%E2\%80\%A2\%22\%2C1\%2C1\%2C1\%2C1\%2C1\%2C1\%2C\%22\%E2\%80\%A2\%22\%2C\%22X\%22\%5D\%2C\%5B1\%2C1\%2C1\%2C1\%2C\%22\%E2\%80\%A2\%22\%2C\%22X\%22\%5D\%2C\%5B1\%2C1\%2C1\%2C1\%2C\%22\%E2\%80\%A2\%22\%2C1\%2C\%22\%E2\%80\%A2\%22\%2C\%22X\%22\%5D\%2C\%5B\%22\%E2\%97\%A6\%22\%2C1\%2C1\%2C1\%2C1\%2C1\%2C1\%2C\%22\%E2\%80\%A2\%22\%2C\%22X\%22\%5D\%2C\%5B\%22Chance4\%22\%2C1\%2C1\%2C1\%2C\%22Chance\%22\%2C\%22Chance\%22\%2C1\%2C1\%2C\%22\%E2\%80\%A2\%22\%5D\%2C\%5B1\%2C1\%2C1\%2C1\%2C1\%2C1\%2C1\%2C\%22\%E2\%80\%A2\%22\%2C\%22X\%22\%5D\%2C\%5B\%22Chance4\%22\%2C1\%2C1\%2C1\%2C\%22Chance\%22\%2C\%22Chance\%22\%2C1\%2C1\%2C\%22\%E2\%80\%A2\%22\%5D\%2C\%5B\%22\%E2\%80\%A2\%22\%2C1\%2C1\%2C1\%2C1\%2C1\%2C1\%2C\%22\%E2\%80\%A2\%22\%2C\%22X\%22\%5D\%2C\%5B1\%2C1\%2C1\%2C1\%2C1\%2C\%22\%E2\%80\%A2\%22\%2C\%22\%E2\%80\%A2\%22\%2C\%22X\%22\%5D\%2C\%5B1\%2C\%22\%E2\%80\%A2\%22\%2C1\%2C1\%2C\%22\%E2\%80\%A2\%22\%2C\%22X\%22\%5D\%2C\%5B1\%2C1\%2C1\%2C1\%2C\%22X\%22\%5D\%2C\%5B1\%2C1\%2C1\%2C\%22\%E2\%80\%A2\%22\%2C\%22X\%22\%5D\%2C\%5B1\%2C1\%2C\%22\%E2\%80\%A2\%22\%2C1\%2C\%22X\%22\%5D\%2C\%5B1\%2C1\%2C1\%2C1\%2C1\%2C\%22\%E2\%97\%A6\%22\%2C\%22X\%22\%5D\%2C\%5B1\%2C\%22\%E2\%97\%A6\%22\%2C1\%2C1\%2C\%22\%E2\%97\%A6\%22\%2C\%22X\%22\%5D\%2C\%5B1\%2C1\%2C1\%2C1\%2C1\%2C\%22\%E2\%80\%A2\%22\%2C\%22\%E2\%80\%A2\%22\%2C\%22X\%22\%5D\%2C\%5B\%22\%E2\%97\%A6\%22\%2C1\%2C1\%2C1\%2C1\%2C1\%2C1\%2C\%22\%E2\%80\%A2\%22\%2C\%22X\%22\%5D\%2C\%5B\%22Chance4\%22\%2C1\%2C1\%2C1\%2C\%22Chance\%22\%2C\%22Chance\%22\%2C1\%2C1\%2C\%22\%E2\%80\%A2\%22\%5D\%2C\%5B1\%2C1\%2C1\%2C1\%2C1\%2C1\%2C1\%2C\%22\%E2\%80\%A2\%22\%2C\%22X\%22\%5D\%2C\%5B\%22Chance4\%22\%2C1\%2C1\%2C1\%2C\%22Chance\%22\%2C\%22Chance\%22\%2C1\%2C1\%2C\%22\%E2\%80\%A2\%22\%5D\%2C\%5B\%22\%E2\%80\%A2\%22\%2C1\%2C1\%2C1\%2C1\%2C1\%2C1\%2C\%22\%E2\%80\%A2\%22\%2C\%22X\%22\%5D\%2C\%5B1\%2C1\%2C1\%2C1\%2C\%22\%E2\%97\%A6\%22\%2C\%22X\%22\%5D\%2C\%5B1\%2C1\%2C1\%2C1\%2C\%22\%E2\%97\%A6\%22\%2C1\%2C\%22\%E2\%80\%A2\%22\%2C\%22X\%22\%5D\%2C\%5B\%22\%E2\%97\%A6\%22\%2C1\%2C1\%2C1\%2C1\%2C1\%2C1\%2C\%22\%E2\%80\%A2\%22\%2C\%22X\%22\%5D\%2C\%5B\%22Chance4\%22\%2C1\%2C1\%2C1\%2C\%22Chance\%22\%2C\%22Chance\%22\%2C1\%2C1\%2C\%22\%E2\%80\%A2\%22\%5D\%2C\%5B1\%2C1\%2C1\%2C1\%2C1\%2C1\%2C1\%2C\%22\%E2\%80\%A2\%22\%2C\%22X\%22\%5D\%2C\%5B\%22Chance4\%22\%2C1\%2C1\%2C1\%2C\%22Chance\%22\%2C\%22Chance\%22\%2C1\%2C1\%2C\%22\%E2\%80\%A2\%22\%5D\%2C\%5B\%22\%E2\%80\%A2\%22\%2C1\%2C1\%2C1\%2C1\%2C1\%2C1\%2C\%22\%E2\%80\%A2\%22\%2C\%22X\%22\%5D\%2C\%5B1\%2C1\%2C1\%2C1\%2C1\%2C\%22\%E2\%80\%A2\%22\%2C\%22\%E2\%80\%A2\%22\%2C\%22X\%22\%5D\%2C\%5B1\%2C1\%2C1\%2C\%22\%E2\%97\%A6\%22\%2C\%22X\%22\%5D\%2C\%5B1\%2C1\%2C1\%2C1\%2C\%22\%E2\%97\%A6\%22\%2C\%22X\%22\%5D\%2C\%5B1\%2C\%22\%E2\%80\%A2\%22\%2C1\%2C\%22\%E2\%97\%A6\%22\%2C1\%2C\%22X\%22\%5D\%2C\%5B1\%2C1\%2C1\%2C1\%2C1\%2C\%22\%E2\%80\%A2\%22\%2C\%22\%E2\%80\%A2\%22\%2C\%22X\%22\%5D\%2C\%5B\%22\%E2\%97\%A6\%22\%2C1\%2C1\%2C1\%2C1\%2C1\%2C1\%2C\%22\%E2\%80\%A2\%22\%2C\%22X\%22\%5D\%2C\%5B\%22Chance4\%22\%2C1\%2C1\%2C1\%2C\%22Chance\%22\%2C\%22Chance\%22\%2C1\%2C1\%2C\%22\%E2\%80\%A2\%22\%5D\%2C\%5B1\%2C1\%2C1\%2C1\%2C1\%2C1\%2C1\%2C\%22\%E2\%80\%A2\%22\%2C\%22X\%22\%5D\%2C\%5B\%22Chance4\%22\%2C1\%2C1\%2C1\%2C\%22Chance\%22\%2C\%22Chance\%22\%2C1\%2C1\%2C\%22\%E2\%80\%A2\%22\%5D\%2C\%5B\%22\%E2\%80\%A2\%22\%2C1\%2C1\%2C1\%2C1\%2C1\%2C1\%2C\%22\%E2\%80\%A2\%22\%2C\%22X\%22\%5D\%2C\%5B1\%2C1\%2C1\%2C1\%2C\%22\%E2\%97\%A6\%22\%2C\%22X\%22\%5D\%2C\%5B1\%2C1\%2C1\%2C1\%2C\%22\%E2\%97\%A6\%22\%2C1\%2C\%22\%E2\%80\%A2\%22\%2C\%22X\%22\%5D\%2C\%5B\%22\%E2\%97\%A6\%22\%2C1\%2C1\%2C1\%2C1\%2C1\%2C1\%2C\%22\%E2\%80\%A2\%22\%2C\%22X\%22\%5D\%2C\%5B\%22Chance4\%22\%2C1\%2C1\%2C1\%2C\%22Chance\%22\%2C\%22Chance\%22\%2C1\%2C1\%2C\%22\%E2\%80\%A2\%22\%5D\%2C\%5B1\%2C1\%2C1\%2C1\%2C1\%2C1\%2C1\%2C\%22\%E2\%80\%A2\%22\%2C\%22X\%22\%5D\%2C\%5B\%22Chance4\%22\%2C1\%2C1\%2C1\%2C\%22Chance\%22\%2C\%22Chance\%22\%2C1\%2C1\%2C\%22\%E2\%80\%A2\%22\%5D\%2C\%5B\%22\%E2\%80\%A2\%22\%2C1\%2C1\%2C1\%2C1\%2C1\%2C1\%2C\%22\%E2\%80\%A2\%22\%2C\%22X\%22\%5D\%2C\%5B1\%2C1\%2C1\%2C1\%2C1\%2C\%22\%E2\%80\%A2\%22\%2C\%22\%E2\%80\%A2\%22\%2C\%22X\%22\%5D\%2C\%5B1\%2C1\%2C\%22\%E2\%80\%A2\%22\%2C\%22X\%22\%5D\%2C\%5B1\%2C1\%2C1\%2C\%22\%E2\%80\%A2\%22\%2C\%22X\%22\%5D\%2C\%5B1\%2C1\%2C1\%2C1\%2C\%22\%E2\%97\%A6\%22\%2C\%22X\%22\%5D\%2C\%5B1\%2C\%22\%E2\%80\%A2\%22\%2C1\%2C\%22\%E2\%80\%A2\%22\%2C1\%2C\%22X\%22\%5D\%2C\%5B1\%2C1\%2C\%22\%E2\%80\%A2\%22\%2C\%22X\%22\%5D\%2C\%5B1\%2C1\%2C1\%2C1\%2C1\%2C\%22\%E2\%80\%A2\%22\%2C\%22\%E2\%80\%A2\%22\%2C\%22X\%22\%5D\%2C\%5B\%22\%E2\%97\%A6\%22\%2C1\%2C1\%2C1\%2C1\%2C1\%2C1\%2C\%22\%E2\%80\%A2\%22\%2C\%22X\%22\%5D\%2C\%5B\%22Chance4\%22\%2C1\%2C1\%2C1\%2C\%22Chance\%22\%2C\%22Chance\%22\%2C1\%2C1\%2C\%22\%E2\%80\%A2\%22\%5D\%2C\%5B1\%2C1\%2C1\%2C1\%2C1\%2C1\%2C1\%2C\%22\%E2\%80\%A2\%22\%2C\%22X\%22\%5D\%2C\%5B\%22Chance4\%22\%2C1\%2C1\%2C1\%2C\%22Chance\%22\%2C\%22Chance\%22\%2C1\%2C1\%2C\%22\%E2\%80\%A2\%22\%5D\%2C\%5B\%22\%E2\%80\%A2\%22\%2C1\%2C1\%2C1\%2C1\%2C1\%2C1\%2C\%22\%E2\%80\%A2\%22\%2C\%22X\%22\%5D\%2C\%5B1\%2C1\%2C1\%2C1\%2C\%22\%E2\%97\%A6\%22\%2C\%22X\%22\%5D\%2C\%5B1\%2C1\%2C1\%2C1\%2C\%22\%E2\%97\%A6\%22\%2C1\%2C\%22\%E2\%80\%A2\%22\%2C\%22X\%22\%5D\%2C\%5B\%22\%E2\%97\%A6\%22\%2C1\%2C1\%2C1\%2C1\%2C1\%2C1\%2C\%22\%E2\%80\%A2\%22\%2C\%22X\%22\%5D\%2C\%5B\%22Chance4\%22\%2C1\%2C1\%2C1\%2C\%22Chance\%22\%2C\%22Chance\%22\%2C1\%2C1\%2C\%22\%E2\%80\%A2\%22\%5D\%2C\%5B1\%2C1\%2C1\%2C1\%2C1\%2C1\%2C1\%2C\%22\%E2\%80\%A2\%22\%2C\%22X\%22\%5D\%2C\%5B\%22Chance4\%22\%2C1\%2C1\%2C1\%2C\%22Chance\%22\%2C\%22Chance\%22\%2C1\%2C1\%2C\%22\%E2\%80\%A2\%22\%5D\%2C\%5B\%22\%E2\%80\%A2\%22\%2C1\%2C1\%2C1\%2C1\%2C1\%2C1\%2C\%22\%E2\%80\%A2\%22\%2C\%22X\%22\%5D\%2C\%5B1\%2C1\%2C1\%2C1\%2C1\%2C\%22\%E2\%80\%A2\%22\%2C\%22\%E2\%80\%A2\%22\%2C\%22X\%22\%5D\%2C\%5B1\%2C1\%2C1\%2C\%22\%E2\%80\%A2\%22\%2C\%22X\%22\%5D\%2C\%5B1\%2C1\%2C1\%2C1\%2C\%22\%E2\%97\%A6\%22\%2C\%22X\%22\%5D\%2C\%5B1\%2C\%22\%E2\%80\%A2\%22\%2C1\%2C\%22\%E2\%80\%A2\%22\%2C1\%2C\%22X\%22\%5D\%2C\%5B1\%2C1\%2C1\%2C1\%2C1\%2C\%22\%E2\%80\%A2\%22\%2C\%22\%E2\%80\%A2\%22\%2C\%22X\%22\%5D\%2C\%5B\%22\%E2\%97\%A6\%22\%2C1\%2C1\%2C1\%2C1\%2C1\%2C1\%2C\%22\%E2\%80\%A2\%22\%2C\%22X\%22\%5D\%2C\%5B\%22Chance4\%22\%2C1\%2C1\%2C1\%2C\%22Chance\%22\%2C\%22Chance\%22\%2C1\%2C1\%2C\%22\%E2\%80\%A2\%22\%5D\%2C\%5B1\%2C1\%2C1\%2C1\%2C1\%2C1\%2C1\%2C\%22\%E2\%80\%A2\%22\%2C\%22X\%22\%5D\%2C\%5B\%22Chance4\%22\%2C1\%2C1\%2C1\%2C\%22Chance\%22\%2C\%22Chance\%22\%2C1\%2C1\%2C\%22\%E2\%80\%A2\%22\%5D\%2C\%5B\%22\%E2\%80\%A2\%22\%2C1\%2C1\%2C1\%2C1\%2C1\%2C1\%2C\%22\%E2\%80\%A2\%22\%2C\%22X\%22\%5D\%2C\%5B1\%2C1\%2C1\%2C1\%2C\%22\%E2\%97\%A6\%22\%2C\%22X\%22\%5D\%2C\%5B1\%2C1\%2C1\%2C1\%2C\%22\%E2\%97\%A6\%22\%2C1\%2C\%22\%E2\%80\%A2\%22\%2C\%22X\%22\%5D\%2C\%5B\%22\%E2\%97\%A6\%22\%2C1\%2C1\%2C1\%2C1\%2C1\%2C1\%2C\%22\%E2\%80\%A2\%22\%2C\%22X\%22\%5D\%2C\%5B\%22Chance4\%22\%2C1\%2C1\%2C1\%2C\%22Chance\%22\%2C\%22Chance\%22\%2C1\%2C1\%2C\%22\%E2\%80\%A2\%22\%5D\%2C\%5B1\%2C1\%2C1\%2C1\%2C1\%2C1\%2C1\%2C\%22\%E2\%80\%A2\%22\%2C\%22X\%22\%5D\%2C\%5B\%22Chance4\%22\%2C1\%2C1\%2C1\%2C\%22Chance\%22\%2C\%22Chance\%22\%2C1\%2C1\%2C\%22\%E2\%80\%A2\%22\%5D\%2C\%5B\%22\%E2\%80\%A2\%22\%2C1\%2C1\%2C1\%2C1\%2C1\%2C1\%2C\%22\%E2\%80\%A2\%22\%2C\%22X\%22\%5D\%2C\%5B1\%2C1\%2C1\%2C1\%2C1\%2C\%22\%E2\%80\%A2\%22\%2C\%22\%E2\%80\%A2\%22\%2C\%22X\%22\%5D\%2C\%5B1\%2C\%22\%E2\%80\%A2\%22\%2C1\%2C1\%2C\%22\%E2\%97\%A6\%22\%2C\%22X\%22\%5D\%2C\%5B1\%2C1\%2C\%22\%E2\%80\%A2\%22\%2C\%22\%E2\%80\%A2\%22\%2C\%22X\%22\%5D\%2C\%5B1\%2C1\%2C1\%2C1\%2C\%22\%E2\%80\%A2\%22\%2C\%22\%E2\%97\%A6\%22\%2C\%22X\%22\%5D\%2C\%5B\%22QFT\%E2\%80\%A06\%22\%5D\%2C\%5B\%22Uncounting6\%22\%5D\%5D\%7D}{quirk link} to a circuit where we iterate on $N=32$ elements using $70$ Toffolis ($\simeq2.5N$) via balanced binary trees and conditionally clean qubits.

\subsection{Unary iteration and QROM using dirty ancilla}\label{sec:unary_iteration_dirty}
If we had access to $n$ dirty qubits instead of clean qubits. We can divide the selection register of size $n$ into a top half of size $k$ and a bottom half of size $n - k$. 
To iterate on the top half, we execute $K = 2^k$ $k$-bit Toffoli gates using constructions described in \sec{mcx_constructions}.  
For each of the $K$ iterations on the top half, we execute a perform a dirty QROM read on the bottom $n - k$ qubits using $n - k$ borrowed dirty ancillae. 
Each of these QROM reads has the same tree shape and is of size $N / K$. 
Since we the ancilla qubits this time are borrowed and can be in an unknown state, we also need to perform a round of laddered toggle detection as described in \sec{laddered_toggle_detection}. 
Because all the $K$ QROM trees are of identical shape with differing data elements, we can do the toggle detection via just 1 more QROM read of the same size but where for each leaf node $i$, we load all data elements from the $i$'th leaf nodes of each of the $K$ QROM trees. 

Let $\text{QROMDirtyCost}(M)$ denote the Toffoli cost of performing a QROM read on $M$ data elements using dirty ancillae (this cost will depend on the specific tree structure used – balanced binary or skewed – and the details of the implementation). 
The overall cost of the procedure is 
$$
    (K + 1) \times \text{QROMDirtyCost}(N / K) + \mathcal{O}(K * k)
$$
Setting $K = \sqrt{N}$ gives us a cost of $1.5N + \mathcal{O}(n \sqrt{N})$ for dirty QROM via balanced binary tree and $1.25N + \mathcal{O}(n \sqrt{N})$ for dirty QROM via the skewed trees case. 

Here is a \href{https://algassert.com/quirk#circuit=

\begin{figure}
\begin{adjustwidth}{-2cm}{-2cm}
\begin{subfigure}{0.3\linewidth}
\includegraphics[width=\linewidth]{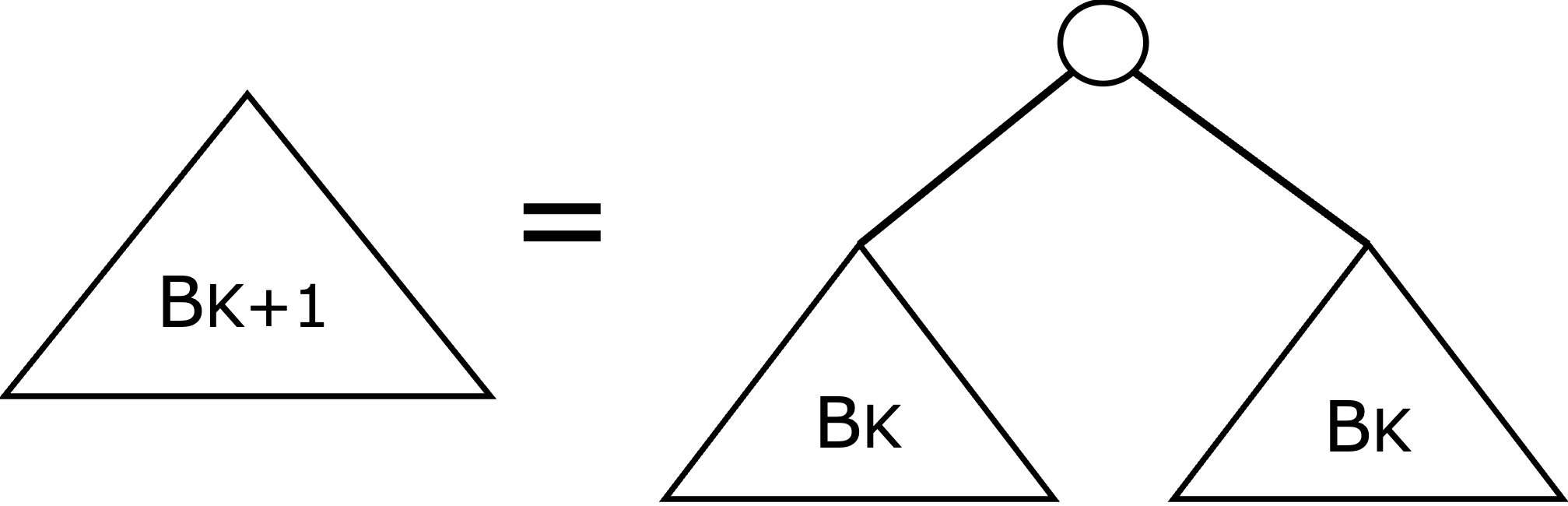}
\caption{Recursive definition of a balanced binary tree of depth $K + 1$. Number of leaf nodes in $B_{K}$ is $2^{K}$ and number of internal nodes is $2^{K} - 1$. 
When $B_{K+1}$ is used for unary iteration, the left subtree corresponds to indices in the range $[0, 2^{K})$ (i.e. $K + 1$'th bit is 0) and the right subtree corresponds to indices in the range $[2^{K}, 2^{K + 1})$ (i.e. $K+1$'th bit is 1). 
Outputs are associated only with the leaf nodes.}
\label{fig:qrom_binary_tree:sfig1}
\end{subfigure}~
\begin{subfigure}{0.6\linewidth}
\includegraphics[width=\linewidth]{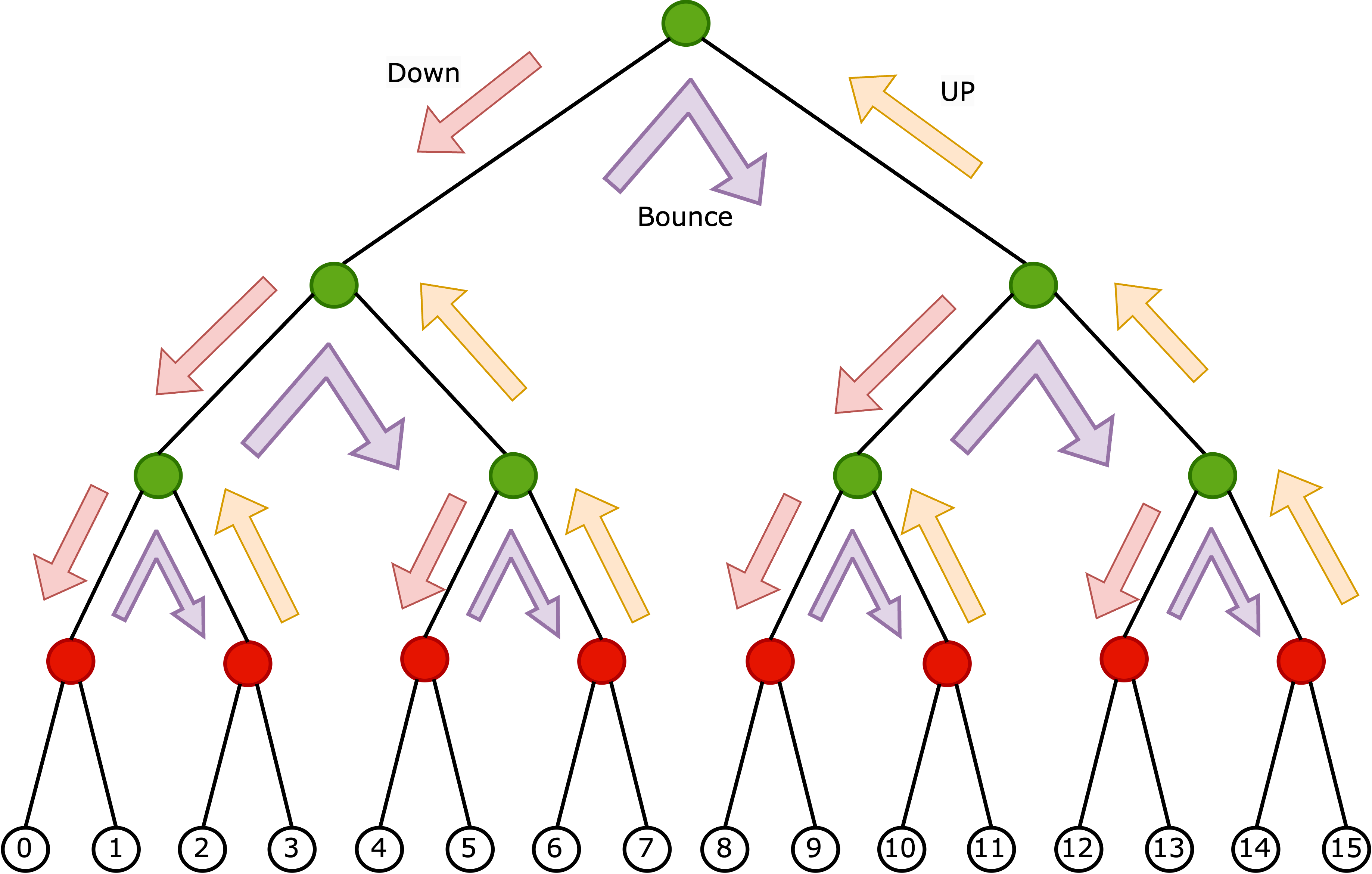}
\caption{
When used for unary iteration on $N$ elements, a balanced binary tree has $N/2$ leaf nodes (marked in red) and $N/2 - 1$ internal nodes (marked in green). A DFS traversal of the tree yields $N / 2 - 1$ DOWN moves, $N / 2 - 1$ BOUNCE moves and $N / 2 - 1$ UP moves. For controlled unary iteration, the number of UP / DOWN moves is $N/2$ (an edge comes in to the root node). Each move corresponds to a circuit element as shown in \fig{qrom_binary_tree:sfig3} and \fig{qrom_binary_tree:sfig4}
}
\label{fig:qrom_binary_tree:sfig2}
\end{subfigure}
\par\bigskip
\begin{subfigure}{0.49\linewidth}
\includegraphics[width=\linewidth, left]{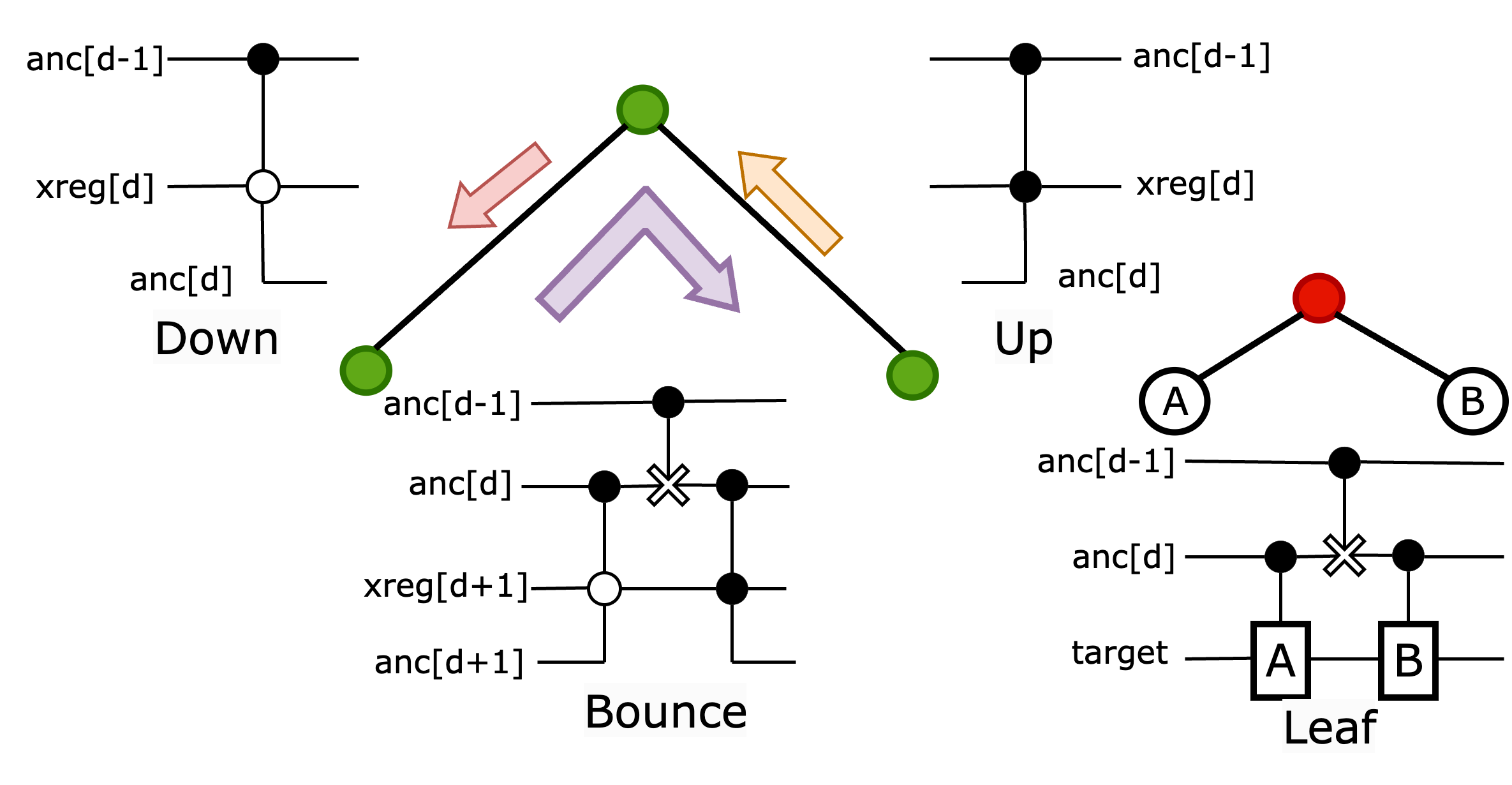}
\caption{Circuit elements produced during tree traversal of a balanced binary tree with clean ancilla. Every DOWN and BOUNCE traversal as a Toffoli cost of 1. 
UP traversal has Toffoli cost of 0 since the $\text{AND}$ gate can be uncomputed using measurement based uncomputation~\cite{Gidney2018halvingcostof}.
LEAF traversal to consume data has a Toffoli cost of 0. 
Thus, Controlled Unary iteration over $N=2^{n}$ elements using $n$ clean ancilla has a Toffoli cost of $N - 1$. This yields the unary iteration construction from~\cite{Babbush2018}}
\label{fig:qrom_binary_tree:sfig3}
\end{subfigure}~
\begin{subfigure}{0.49\linewidth}
\includegraphics[width=\linewidth, right]{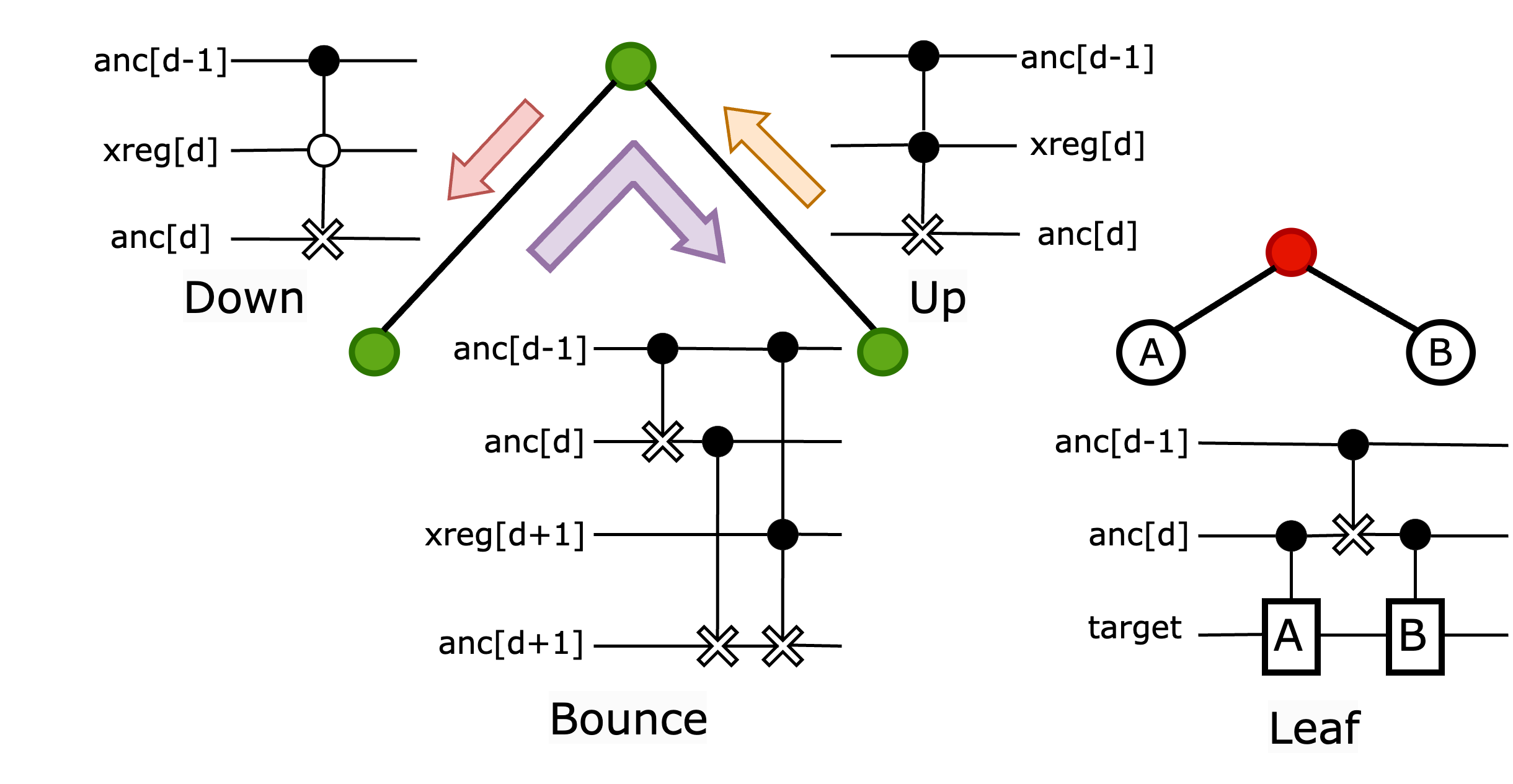}
\caption{Circuit elements produced during tree traversal of a balanced binary tree with (potentially dirty) ancilla where measurement based uncomputation is not allowed. 
Every DOWN, BOUNCE and UP traversal now has a Toffoli cost of 1.
LEAF traversal to consume data has a Toffoli cost of 0. 
Controlled Unary iteration over $N=2^{n}$ elements using $n$ ancilla (without measurement based uncompuation) has a Toffoli cost of $1.5N - 1$. This yields the unary iteration construction from~\cite{Childs2018}}
\label{fig:qrom_binary_tree:sfig4}
\end{subfigure}
\begin{subfigure}{\linewidth}
\includegraphics[width=\linewidth]
{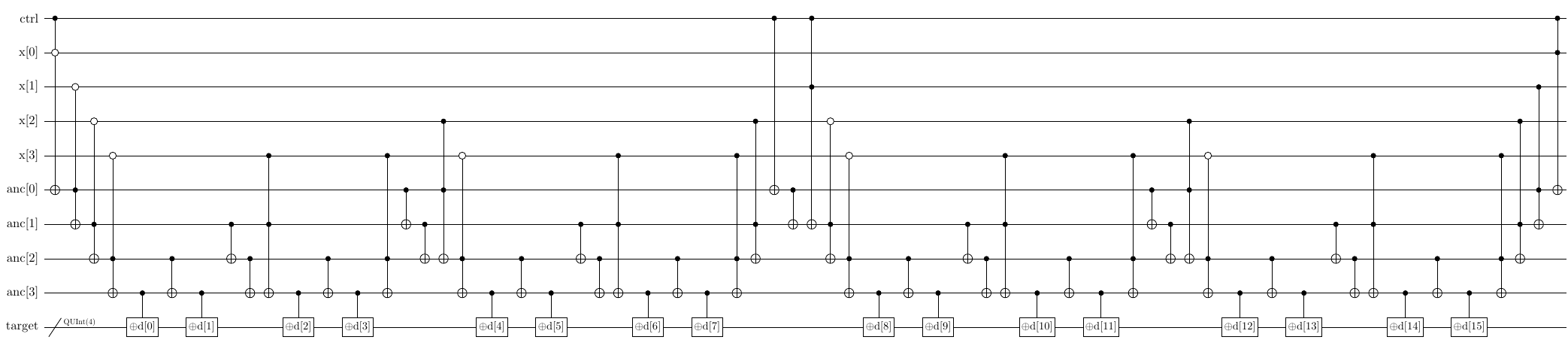}
\caption{QROM circuit to load $N=16$ data elements using $4$ ancilla qubits (without measurement based uncomputation) and $1.5N - 1 = 23$ Toffoli gates. 
Uses circuit elements described in \fig{qrom_binary_tree:sfig4} and forms the basis of (i) unary iteration using only 2 clean ancilla and $2.5N$ Toffoli AND (ii) unary iteartion using $n=\log_2{N}$ dirty ancilla and $1.5N + \mathcal{O}(n\sqrt{N})$ Toffoli}
\label{fig:qrom_binary_tree:sfig5}
\end{subfigure}

\caption{Unary iteration using balanced binary trees}
\label{fig:qrom_binary_tree}
\end{adjustwidth}
\end{figure}

\begin{figure}
\begin{adjustwidth}{-2cm}{-2cm}
\begin{subfigure}{0.3\linewidth}
\includegraphics[width=\linewidth]{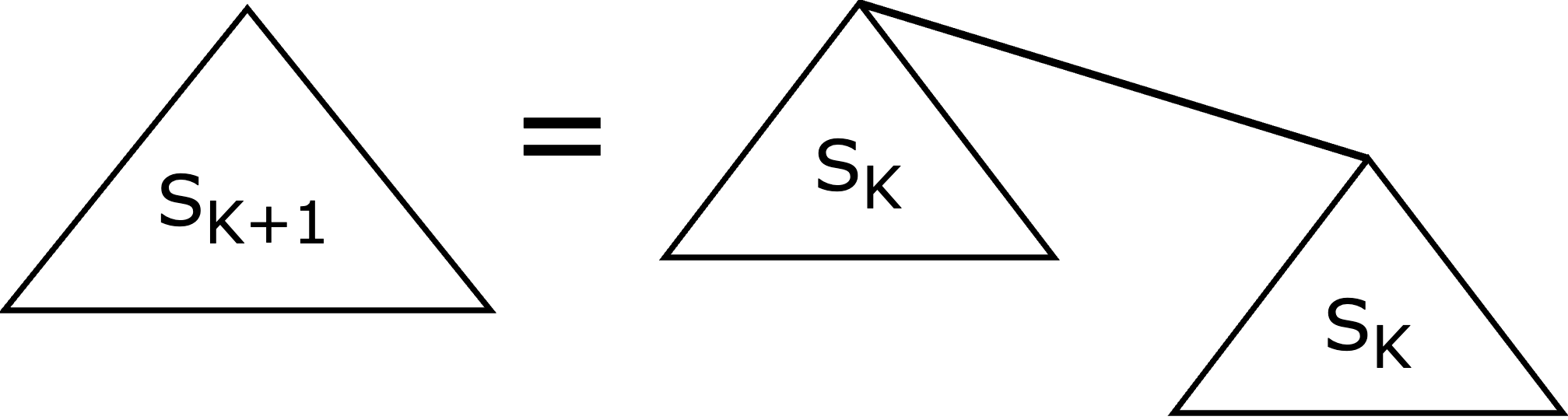}
\caption{Recursive definition of a Skew tree of depth $K + 1$. The size of $S_{k}$ is $2^{k}$, thus when doing unary iteration a skew tree associates an output with every node. The basic idea of using skewed trees for unary iteration is to replace circuits of the form ``if $\lnot C$ then do $A$; if $C$ then do $B$" with circuits of the form ``do $A$; if $C$ then do $A^{-1} \cdot B$". Instead of doing $A$ conditionally, we do $A$ unconditionally but undo $A$ in addition to doing $B$ when $C$ is true. This optimization works best when $A^{-1} \cdot B$ is efficient to apply, which is the case for $\text{QROM}$}
\label{fig:qrom_skew_tree:sfig1}
\end{subfigure}~
\begin{subfigure}{0.6\linewidth}
\includegraphics[width=\linewidth]{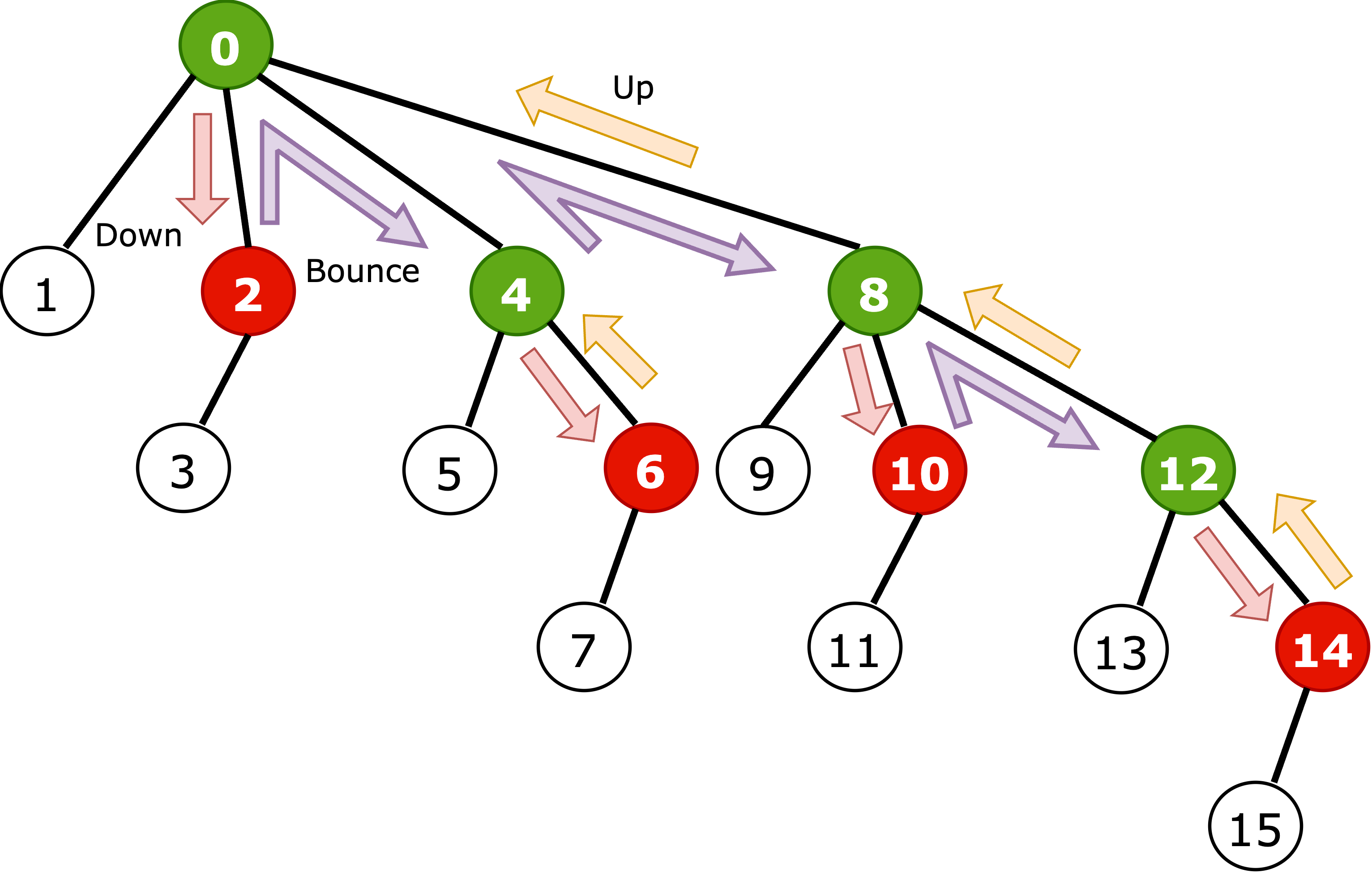}
\caption{When used for unary iteration on $N$ elements, a skew tree has $N/4$ leaf nodes (marked in red) and $N/4$ internal nodes (marked in green). A DFS traversal of the tree yields $N/4$ DOWN moves, $N/4 - 1$ BOUNCE moves and $N/4$ UP moves. Each move corresponds to a circuit element as shown in \fig{qrom_skew_tree:sfig3} and \fig{qrom_skew_tree:sfig4}.
}
\label{fig:qrom_skew_tree:sfig2}
\end{subfigure}
\par\bigskip
\begin{subfigure}{0.49\linewidth}
\includegraphics[width=\linewidth, left]{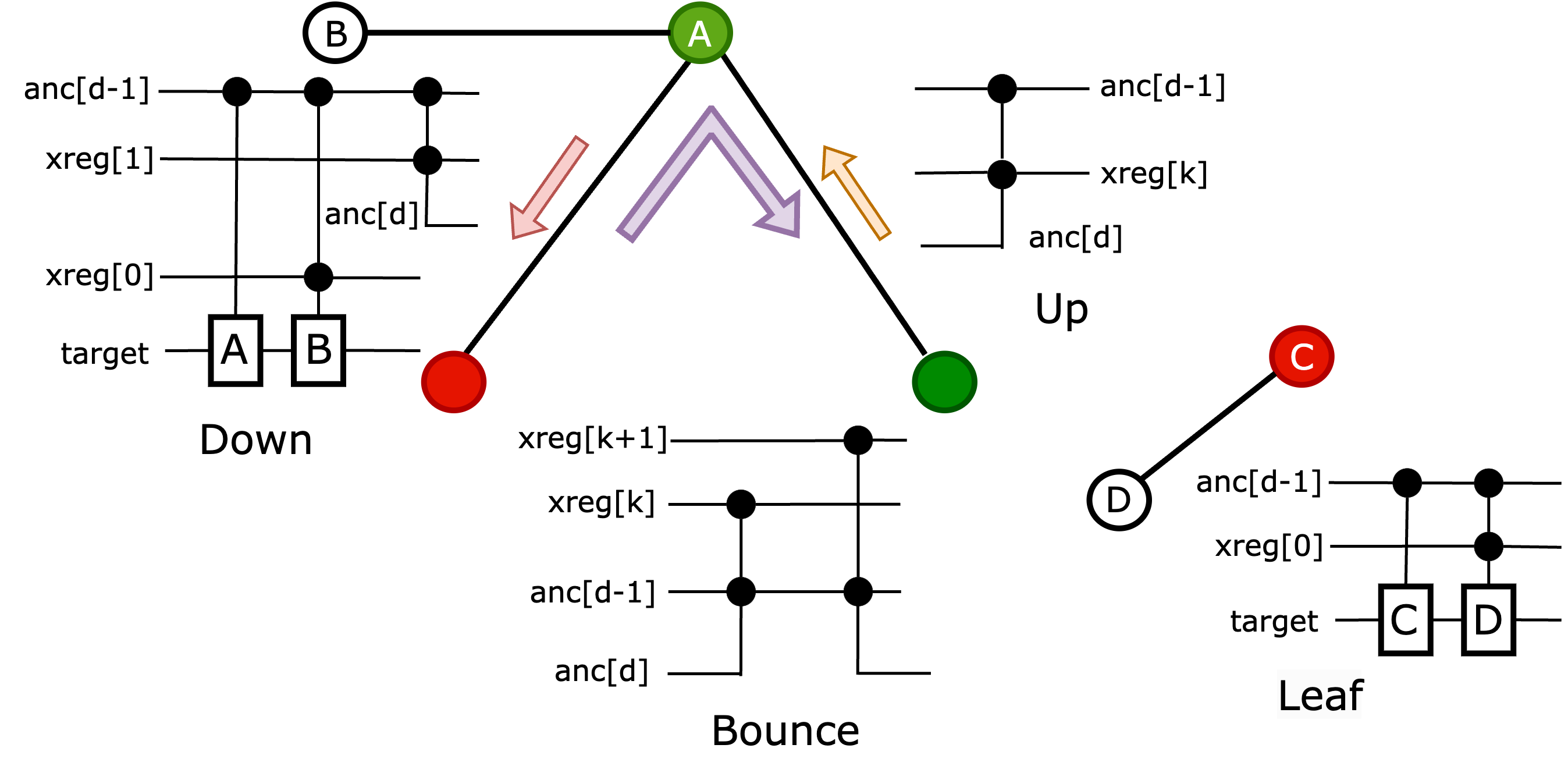}
\caption{Circuit elements produced during tree traversal of a skew tree with clean ancilla. DOWN traversal has a Toffoli cost of 2, BOUNCE traversal has a Toffoli cost of 1, UP traversal has a Toffoli cost of 0 since the $\text{AND}$ gate can be uncomputed using measurement based uncomputation~\cite{Gidney2018halvingcostof}.
LEAF traversal to consume data has a Toffoli cost of 1. 
Thus, Controlled Unary iteration over $N=2^{n}$ elements using $n$ clean ancilla has a Toffoli cost of $N - 1$.}
\label{fig:qrom_skew_tree:sfig3}
\end{subfigure}~
\begin{subfigure}{0.49\linewidth}
\includegraphics[width=\linewidth, right]{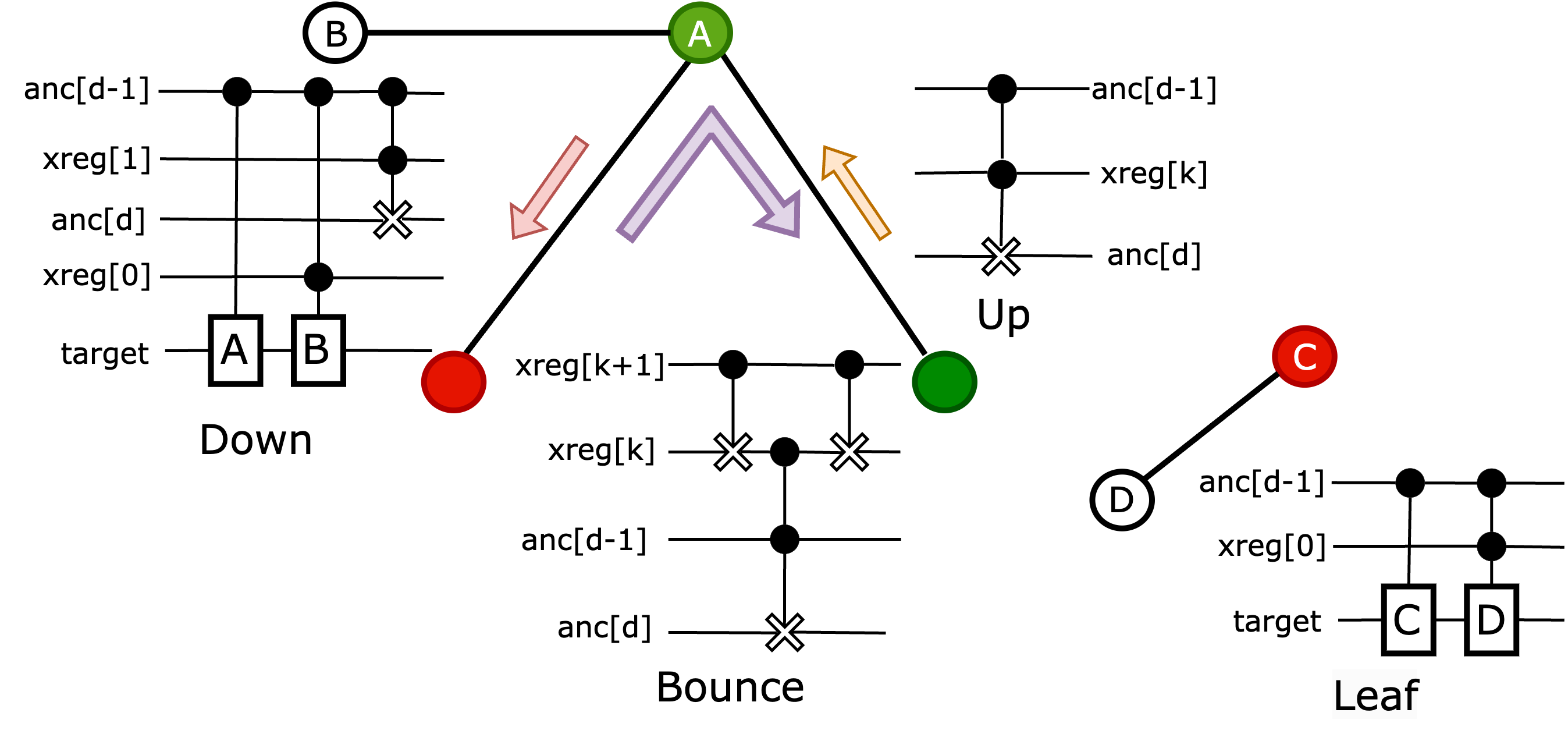}
\caption{Circuit elements produced during tree traversal of a skew tree with (potentially dirty) ancilla where measurement based uncomputation is not allowed. 
DOWN traversal has a Toffoli cost of 2. BOUNCE, UP and LEAF traversals have a Toffoli cost of 1. 
Thus, Controlled Unary iteration over $N=2^{n}$ elements using $n$ ancilla (without measurement based uncompuation) using Skew trees has a Toffoli cost of $\frac{5}{4}N - 1 = 1.25N - 1$.}
\label{fig:qrom_skew_tree:sfig4}
\end{subfigure}
\caption{Optimized Unary iteration using skew trees. The optimization can be applied only when $A^{-1} \cdot B$ is efficient to apply where $A$ and $B$ are multiplexed unitaries applied via unary iteration. For example - $\text{QROM}$ satisfies this criteria.}
\label{fig:qrom_skew_tree}
\end{adjustwidth}
\end{figure}

\section{Conclusion}\label{sec:conclusion}
In this work, we have demonstrated the utility of \emph{conditionally clean ancilla} qubits as a powerful tool for constructing efficient quantum circuits.  By formally defining this concept and providing a general framework for its application, we have shown how to leverage conditionally clean ancillae to achieve significant improvements in gate count and depth for several fundamental quantum primitives.

We introduced \emph{laddered toggle detection}, an improved technique for replacing clean ancillae with dirty ancillae, mitigating the exponential overhead that can arise from naive toggle detection strategies. We then presented novel circuit constructions for the $n$-bit Toffoli gate, the $n$-bit Incrementer, the $n$-bit Quantum-Classical Comparator, and Unary Iteration (including QROM).  These constructions improve upon the previously best known constructions and achieve the lowest Toffoli gate counts in the sublinear ancilla regime.

When compared to constructions using a linear number of ancilla qubits, our constructions have a higher number of T / Toffoli gate counts because uncomputing a clean ancilla qubit can often be done with only measurement + clifford operations and requires no T / Toffoli gates. However, we do not have a way to cheaply uncompute the conditionally clean ancilla qubits we use as part of our constructions. This gives us a 2x overhead in terms of T / Toffoli gate counts. It's an open question to figure out whether one can extend the ideas of measurement based uncomputation to cheaply uncompute the conditionally clean ancilla qubits, similar to the clean ancilla case. 

We believe that conditionally clean ancillae should become a standard tool in the quantum circuit designer's toolkit and will be useful for circuit optimizations beyond the specific circuit constructions we provide.

\section{Contributions}

Craig guided the project and came up with some initial constructions.
Tanuj improved the constructions, found additional ones, and wrote the paper as well as the accompanying code.

\section{Acknowledgements}
We thank Robin Kothari for helpful discussions on the abstract computer science problem presented in \problemref{problem_1}. We thank Hartmut Neven for creating an environment where this research was possible.

\clearpage 

\printbibliography

\end{document}